\documentclass[a4paper,12pt]{article}
\usepackage{indentfirst}
\usepackage{graphicx}
\usepackage{graphics}
\usepackage{epstopdf}
\usepackage{epsfig}
\usepackage{overpic}
\usepackage{array}
\usepackage{color}
\usepackage{multirow}
\usepackage{float}
\usepackage{subfigure}
\usepackage{amsmath,amssymb, amsthm}
\usepackage[mathscr]{euscript}
\usepackage{latexsym,bm}
\usepackage{cite}
\usepackage{booktabs}
\usepackage{bbm}
\usepackage{diagbox}
 \usepackage[all]{xy}
 \usepackage{color}
\usepackage{algorithm}  
\usepackage{algorithmicx} 
\usepackage{appendix}  

\newcommand\D{\textup{d}}

{\theoremstyle{remark} }
\newtheorem{proposition}{Proposition}
\def\mi{\mathbbm{i}}
\def\me{\mathbbm{e}}

\makeatletter
\@addtoreset{equation}{section}
\makeatother

\begin{document}
\bibliographystyle{unsrt}

\title{Numerical methods for the Wigner equation with unbounded potential}
\author{Zhenzhu Chen\footnotemark[2],
\and Yunfeng Xiong\footnotemark[2],
\and Sihong Shao\footnotemark[2] $^,$\footnotemark[1]}

\renewcommand{\thefootnote}{\fnsymbol{footnote}}
\footnotetext[2]{LMAM and School of Mathematical Sciences, Peking University, Beijing 100871, China.}
\footnotetext[1]{To
whom correspondence should be addressed. Email:
\texttt{sihong@math.pku.edu.cn}}
\date{\today}
\maketitle

\begin{abstract}

Unbounded potentials are always utilized to strictly confine quantum dynamics and generate bound or stationary states due to the existence of quantum tunneling. However, the existed accurate Wigner solvers are often designed for either localized potentials or those of the polynomial type. This paper attempts to solve the time-dependent Wigner equation in the presence of a general class of unbounded potentials by exploiting two equivalent forms of the pseudo-differential operator: integral form and series form (i.e., the Moyal expansion). The unbounded parts at infinities are approximated or modeled by polynomials and then a remaining localized potential dominates the central area. The fact that the Moyal expansion reduces to a finite series for polynomial potentials is fully utilized.  Using a spectral collocation discretization which conserves both mass and energy, several typical quantum systems are simulated with a high accuracy and reliable estimation of macroscopically measurable quantities is thus obtained.

\vspace*{4mm}


\end{abstract}

\section{Introduction}
\label{sec:intro}

Unbounded potentials are ubiquitous in quantum mechanics, especially in simulating quantum tunneling phenomena ranging from various branches of physics and chemistry. As a typical example among them, the double-well potentials with two minima separated by a barrier have been widely used in understanding the transition of quantum states \cite{KierigSchnorrbergerSchietingerTomkovicOberthaler2008,WeinerTse1981} as well as in modeling the potential energy surface of small molecules like the ammonia and the methane \cite{KaShin2003,bk:Pilar2013}. The often-used 
double-well potentials can be simply characterized by unbounded polynomial potentials plus bounded localized potentials, and related quantum observables are calculated through either the Schr\"{o}dinger equation \cite{KaShin2003} or the path integral formalism \cite{PerezTuckermanMuser2009}. In this work, we turn to adopt the Wigner function \cite{Wigner1932}, a quasi-probability distribution, to investigate the quantum tunneling effects, because it grants us a natural description of quantum observables in a statistical form due to the Weyl correspondence \cite{tatarskiui1983,KluksdahKrimanlFerryRinghofer1989}.

However, solving the Wigner equation that describes the time evolution of the Wigner function in the phase space
is usually a tough task because of the difficulties in tackling the nonlocal and highly oscillating pseudo-differential operator.
The situation becomes worse when the unbounded potentials are taken into account. 
The existed deterministic solvers, including the finite difference schemes \cite{Frensley1990,th:Biegel1997} and the spectral collocation methods \cite{Ringhofer1990,ShaoLuCai2011,XiongChenShao2016}, 
always require the potentials to decay fast and vanish at infinities, i.e., the localized potentials, 
since the Wigner kernel is evaluated by the Poisson summation formula, 
and thus is not applicable for the unbounded potentials. 
For the potentials of the polynomial type, on the other hand, 
as an equivalent series form of the pseudo-differential operator,
the Moyal expansion reduces to a finite series  \cite{bk:MarkowichRinghoferSchmeiser1990,bk:Schleich2011} and the resulting equation can be solved by either the spectral method \cite{ThomannBorzi2016} or the Hermite expansion \cite{FurtmaierSucciMendoza2015}. In this work, we attempt to combine the advantages of the above two to evolve the Wigner quantum dynamics in the presence of unbounded potentials.  

\begin{figure}
    \centering
    \includegraphics[width=1.0\textwidth, height=0.25\textwidth]{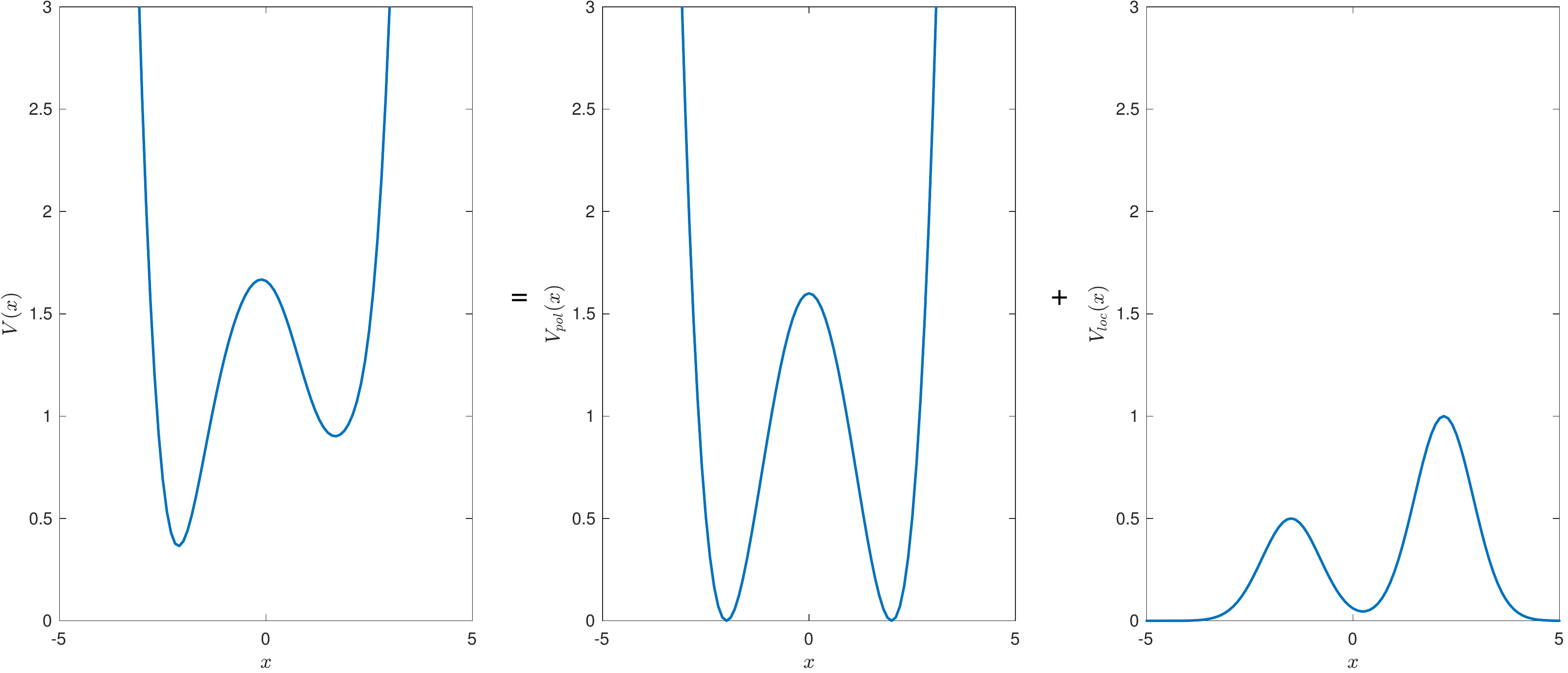}
    \caption{{\small An asymmetric double-well potential $V(x) = 0.1(x^2-2^2)^2+\me^{-(x-2.2)^2}+0.5\me^{-(x+1.5)^2}$ \cite{PerezTuckermanMuser2009}. {It} can be decomposed into a polynomial part $V_{pol}(x) = 0.1(x^2-2^2)^2$} and a localized part $V_{loc}(x)= \me^{-(x-2.2)^2} + 0.5\me^{-(x+1.5)^2}$.}
    \label{fig:potential_split}
\end{figure}

Specifically, we focus on the potentials $V(x)$ that are sufficiently smooth, 
with the asymptotic behaviors at infinities governed or approximated by polynomials. 
Such kind of potentials covers the {double-well potentials} as mentioned above 
and is commonly used to fit the observed data (known as the polynomial regression model). 
For simplicity, we assume that the unbounded $V(x)$ can be split into a polynomial potential $V_{pol}(x)$ and a localized one $V_{loc}(x)$ that decays at infinities, e.g., as displayed in Fig. \ref{fig:potential_split}. 
In this manner, two equivalent forms: the integral form and the Moyal expansion, can be employed to deal with $V_{loc}(x)$ and $V_{pol}(x)$ separately and the resulting equation can be solved by different techniques. 
In this paper, we choose the spectral collocation method because it is able to accurately resolve both the linear differential operators and the Fourier integrals for sufficiently smooth potentials. {Meanwhile, discussions on the conservation of both mass and energy, as well as on their numerical counterparts, are performed. {In some situations, the operator splitting technique can be further introduced to improve the performance.}

{The proposed scheme allows us to study the dynamics of many interesting quantum systems, like the P\"oschl-Teller potential \cite{BundTijero2000,bk:Flugge2012} and the double-well {systems} \cite{FurtmaierSucciMendoza2015,KaczorKlimasSzydlowskiWoloszynSpasak2016}.}
Several macroscopically measurable quantities, 
such as the differences in energy levels, 
the quantum tunneling rate, and the autocorrelation function can also be obtained with a satisfactory accuracy. 
Moreover, the uncertainty principle will be shown directly in phase space 
by simulating the superposition of a harmonic oscillator perturbed by a sixth-order anharmonic one.

The remaining is organized as follows. Section \ref{sec:2} gives a brief introduction to the Wigner equation with two equivalent forms.  In Section \ref{sec:3}, the treatment of unbounded potentials as well as the  numerical scheme for the resulting equation is presented. 
We will prove that the proposed scheme can {maintain the mass and total energy}, and the demonstration of {its} performance is left for Section \ref{sec:4}. Concluding remarks and further discussions are collected in the Section \ref{sec:5}.

\section{Quantum mechanics in phase space}
\label{sec:2}

As a classical mathematical representation of quantum mechanics in phase space, 
the Wigner function allows a direct connection with the classical picture and its dynamics equation reduces to the classical Vlasov equation as the reduced {Planck} constant $\hbar$ vanishes. In this section, we will sketch the Wigner formalism for quantum mechanics and the exposition is restricted to one-dimensional one-body situation for simplicity. 
The Wigner function $f(x,k,t)$ is defined in the phase space $(x,k)\in\mathbb{R}^2$ for the position $x$ and the wavenumber $k$ through the Weyl-Wigner transform of the density matrix $\rho(x,x',t)$ as follows
\begin{equation}
  \label{eq:wigner_function}
  f(x,k,t) = \int_{-\infty}^{+\infty}\me^{-\mi ky}\rho(x+{y}/{2},x-{y}/{2},t)\D y.
\end{equation}
If there are $N$ stationary states $\{\phi_n(x)\}$ corresponding to energy $\{E_n\}$ 
with $n=1,2,\ldots,N$, then the dynamics is given by $\psi_n(x,t) := \phi_n(x)\me^{-\mi E_nt/\hbar}$
and thus the {superposed} state of these $N$ states is $\Psi(x,t) := \sum_{n=1}^{N}a_n\psi_n(x,t)$.
In consequence, the Wigner function for such {superposed} state reads
\begin{equation}
  \label{eq:wigner_fun_pure}
f(x,k,t)=\sum_{n,m=1}^{N} a_na_m \me^{\mi(E_m-E_n)t/\hbar} \int_{-\infty}^{+\infty}\me^{-\mi ky}\phi_n(x+{y}/{2})\phi_m^*(x-{y}/{2})\D y .
\end{equation}
Given a quantum operator $\hat{A}(x,k)$ at the instant $t$, the expectation value can be calculated by averaging the corresponding Weyl symbol $A(x, k)$ with the Wigner function: 
\begin{equation}
  \label{eq:expectation}
  \langle \hat{A}\rangle(t) = \int_{-\infty}^{+\infty}\int_{-\infty}^{+\infty}A(x,k)f(x,k,t)\D x\D k.
\end{equation}
Therefore we can easily deduce from Eqs.~\eqref{eq:wigner_fun_pure} and \eqref{eq:expectation} that 
$\langle \hat{A} \rangle(t)$ can be decomposed into {$\frac{N(N-1)}{2}+1$ components,  the frequencies of which} must be proportional to $\left|E_m-E_n\right|$ for all $m,n=1,2,\cdots,N$.
{That is, it is readily to obtain the differences in energy levels of the quantum system in question only via a direct spectrum analysis of $\langle \hat{A}\rangle(t)$.}

Starting from the Schr\"{o}dinger equation, 
it can be shown that the Wigner function $f(x,k,t)$ follows the following {dynamics}, 
i.e., the time-dependent Wigner equation,
\begin{align}
  \label{eq:wigner_original}
  \frac{\partial}{\partial t}f(x,k,t)+&\frac{\hbar k}{m} \nabla_x f(x,k,t) = \Theta_V[f](x,k,t),  \\
  \label{eq:pseudo_differential}
  \Theta_V[f](x,k,t)&=\frac{1}{\mi\hbar}\int_{-\infty}^{+\infty}\me^{-\mi ky}\left(V(x+{y}/{2})-V(x-{y}/{2})\right)\hat{f}(x,y,t)\D y,\\
     \hat{f}(x,y,t) &= \frac{1}{2\pi}\int_{-\infty}^{+\infty}\me^{\mi ky}f(x,k,t)\D k,
\end{align}
where $m$ is the mass and $\hat{f}(x,y,t)$ is nothing but $\rho(x+{y}/{2},x-{y}/{2},t)$ from Eq.~\eqref{eq:wigner_function} 
through an inverse Fourier transform. Here $\Theta_V$ is the so-called nonlocal pseudo-differential operator containing the quantum information and has different but equivalent expressions as follows. 
\begin{itemize}
  \item For {$V(x)\in L(\mathbb{R}) $, the pseudo-differential operator is characterized by a convolution one:}
 \begin{align}
 \label{eq:theta_local}
      \Theta_V[f](x,k,t)&=\int_{-\infty}^{+\infty}f(x,k',t)V_w\left(x,k-k'\right)\D k', \\
      V_w(x,k)&=\frac{1}{2\pi\mi\hbar}\int_{-\infty}^{+\infty}\me^{-\mi ky}\left(V(x+{y}/{2})-V(x-{y}/{2})\right)\D y, \label{eq:kernel}
\end{align}
where $V_w(x,k)$ is the so-called Wigner potential or Wigner kernel. 
\item  For $V(x)\in C^{\omega}(\mathbb{R})$, performing the Taylor series for $V(x\pm y/2)$ at {$x$} yields
\begin{equation}
    \label{eq:Taylor}
    V(x+{y}/{2})-V(x-{y}/{2})=\sum_{l=0}^{+\infty}\frac{y^{2l+1}}{2^{2l}(2l+1)!} \nabla_x^{2l+1}V(x),
\end{equation}
and substituting the above expression \eqref{eq:Taylor} into Eq.~\eqref{eq:pseudo_differential} leads to
the Moyal expansion
\begin{align}
     \Theta_V[f](x,k,t) &= \sum_{l=0}^{+\infty}\xi_l(x) \nabla_{k}^{2l+1} f(x,k,t), \label{eq:theta_poly}\\
     \xi_l(x) &= \frac{(-1)^l}{2^{2l}(2l+1)!\hbar}\nabla_x^{2l+1}V(x). \label{eq:theta_poly_xi}
\end{align}
Here we adopt the compact notations: {$\nabla_x^n = {\partial^n}/{\partial x^n}$
and $\nabla_k^n = {\partial^n}/{\partial k^n}$} for  $n\in\mathbb{N}$. 
\item  For $V(x)=V_{loc}(x)+V_{ub}(x)$ with {$V_{loc}(x)\in L(\mathbb{R})$} and $V_{ub}(x)\in C^\omega(\mathbb{R})$,  representing a general class of (unbounded) potentials {which can be decomposed into two parts (e.g., see Fig.~\ref{fig:potential_split}). Owing to the linearity of the Fourier transform, the pseudo-differential operator $\Theta_V$ can be rewritten as a linear combination of Eqs.~\eqref{eq:theta_local} and \eqref{eq:theta_poly}.}
More importantly, if $V_{ub}(x)$ is indeed a polynomial or can be approximated at infinities by a polynomial, denoted by $V_{pol}(x)$, 
then we only need to consider a much simpler expression: 
\begin{equation}
\label{eq:key1}
 \Theta_V[f](x,k,t)= \int_{-\infty}^{+\infty}V_{w}^{loc}(x,k-k')f(x,k',t)\D k'\\
   +\sum_{l=0}^{[P/2]}\xi_l^{pol}(x)\nabla_k^{2l+1}f(x,k,t),
\end{equation}
where $V_{w}^{loc}(x,k)$ is the Wigner kernel corresponding to the localized potential $V_{loc}(x)$ via Eq.~\eqref{eq:kernel}, $\xi_l^{pol}(x)$ give the coefficients from the polynomial potential $V_{pol}(x)$ via Eq.~\eqref{eq:theta_poly_xi}, and $P$ denotes the degree. 
A key observation in Eq.~\eqref{eq:key1} is 
the Moyal expansion reduces to a finite series which can be readily resolved by standard numerical techniques.
\end{itemize} 

 
A detailed comparison of the above three expressions \eqref{eq:theta_local}, \eqref{eq:theta_poly}, and \eqref{eq:key1} for the {pseudo-differential operator} shows: {Eq.~\eqref{eq:theta_local} fails to hold when $V(x)$ is unbounded, while Eq.~\eqref{eq:theta_poly} involves infinite terms even when $V(x) \in C^{w}(\mathbb{R})$ with a compact support. By contrast, Eq.~\eqref{eq:key1}, albeit in a somewhat complicated form, only requires the asymptotic behavior of $V(x)$ at infinities is governed by polynomials. It deals with the unbounded part with a finite series of linear differential term, and captures the fine structure in the central area through a twisted convolution. A simple example is the following unbounded rational-fraction-type potential }
\begin{equation}
  \label{eq:fraction}
  V(x) = \frac{x^4+1}{x^2+1},
\end{equation}
which converges to $x^2$ as $x \to \infty$. {It} can be easily verified that, 
the expression \eqref{eq:theta_local} breaks down because the Wigner kernel is not well defined any more in the classical sense, and the infinite series in the Moyal expansion \eqref{eq:theta_poly} sticks in there and thus is difficult to handle with; on the contrary, the expression \eqref{eq:key1} gets rid of those problems by taking 
\begin{equation}
\label{eq:split1}
V_{pol}(x) = x^2-1, \quad V_{loc}(x) = V(x) - V_{pol}(x) =\frac{2}{x^2+1}. 
\end{equation}
In summary, the combined expression \eqref{eq:key1} {serves as} the starting point of this work for investigating the Wigner quantum dynamics in the presence of unbounded potentials.

{Before proceeding, we would like to mention two {conservation laws} that are always used to guide the design of numerical methods.} One is the mass conservation stated by 
the continuity equation
\begin{equation}
\label{eq:continuity}
  \frac{\partial}{\partial t}n(x,t)+\nabla_xj(x,t) = 0,
\end{equation}
where $n(x,t)$ is the particle density and $j(x,t)$ the current density \cite{XiongChenShao2016}.  
The other is the energy conservation
\begin{equation} \label{eq:energy_exp}
\frac{\D}{\D t} \langle \hat{H}\rangle = 0 \,\,\, \text{with} \,\,\, \hat{H} = \frac{\hat{p}^2}{2m} + V(\hat{x}), 
\end{equation}
where $\hat{H}$ is the quantum Hamiltonian operator.

\section{Numerical methods}
\label{sec:3}

Now we turn to seek a numerical approximation to the Wigner equation, where the {nonlocal} term has the form as in Eq.~\eqref{eq:key1}. The convolution term poses the first challenge since it involves double integrations.  In general, a simple nullification of the distribution outside a sufficiently large $k$-domain is usually adopted \cite{ShaoLuCai2011,XiongChenShao2016}.  For a sufficiently large $k$-domain 
$\mathcal{K}=[k_{\min}, k_{\max}]$, 
the truncated version of the Wigner equation is
\begin{equation}
  \label{eq:truncated_W}
  \begin{split}
    \frac{\partial}{\partial t}f(x,k,t)+\frac{\hbar k}{m}\nabla_xf(x,k,t) =&\sum_{l=0}^{[P/2]}\xi_l^{pol}(x)\nabla_k^{2l+1}f(x,k,t)\\
    &+\int_{k_{\min}}^{k_{\max}}\tilde{V}_w^{loc}(x,k-k')f(x,k',t)\D k', \\
   \end{split}
\end{equation}
where
\begin{equation}
    \tilde{V}_w^{loc}(x,k) = \frac{\Delta y}{2\pi\mi\hbar} \sum_{\zeta=-\infty}^{+\infty}\left[V_{loc}(x+{y_\zeta}/{2})-V_{loc}(x-{y_\zeta}/{2})\right] \me^{-\mi k y_\zeta}
\end{equation}
denotes the {discretized} Wigner kernel for the localized potential $V_{loc}(x)$ {and $y_\zeta = \zeta \Delta y$ with $\Delta y$ being the spacing step in $y$-space}.  {Such approximation stems from the Poisson summation formula:
\begin{equation}\label{Poisson_summation_formula}
\sum_{n=-\infty}^{\infty} V_w(x, k+{2\pi n}/{\Delta y}) = \frac{\Delta y}{2\pi\mi\hbar} \sum_{\zeta=-\infty}^{+\infty}\left[V_{loc}(x+{y_\zeta}/{2})-V_{loc}(x-{y_\zeta}/{2})\right] \me^{-\mi k y_\zeta}.
\end{equation}
Here we assume that $V_w$ decays and thus ignore the periodic images. }A necessary and sufficient condition for the truncated Wigner equation \eqref{eq:truncated_W} to conserve the mass
has been given in \cite{XiongChenShao2016} and reads  
\begin{equation}\label{eq:suf_con}
  L_k \Delta y = 2\pi,
\end{equation}
where $L_k = k_{\max}-k_{\min}$ represents the length of $k$-domain. In $x$-space, the popular quantum transitive boundary condition will be adopted hereafter as did in \cite{ShaoLuCai2011}.

The spectral (element) collocation methods have been demonstrated in \cite{ShaoLuCai2011,XiongChenShao2016} to  resolve the oscillations of the Wigner function and thus will be utilized in this work to discretize
the truncated Wigner equation \eqref{eq:truncated_W}. In particular, a Fourier spectral collocation scheme is adopted in $k$-space and a collocation spectral element method with Gauss-Lobbato points in $x$-space. 
An explicit fourth-order Runge-Kutta discretization \cite{GottliebShu1998} is then employed for the time marching as did in \cite{ShaoLuCai2011}. Actually, the plane wave expansion in $k$-space is a natural choice since both the linear differential operator and the convolution term can be accurately approximated in a compact form. Furthermore, we are able to prove the numerical conservation of both mass and energy for the resulting full discretization.

\subsection{The Spectral collocation method}
\label{sec:spectral}

The $N$ uniform collocation points in $\mathcal{K}$ are $k_j = k_{\min} + {jL_k}/{N}$ with $j=0,1,\ldots,N-1$ 
and then the plane wave expansion in $k$-space reads 
\begin{equation}
  \label{eq:approx_k}
  f(x,k,t)\approx  \sum_{\nu=-N/2+1}^{N/2} a_\nu(x,t) \psi_{\nu}(k),
\end{equation}
where 
\begin{equation}
\psi_\nu(k) = \me^{{2\pi} \mi \nu (k-k_{\min})/L_k},
\end{equation}
and the coefficients $\{a_\nu(x,t)\}$ are {further} determined by a collocation spectral element method with Gauss-Lobbato points in $x$-space for easy implementation of boundary conditions. Let $\mathcal{X}=[x_L,x_R]$ be the computational domain in $x$-space
and we divide it into $Q$ non-overlapping elements as $ \mathcal{X} = \bigcup_{q=1}^{Q}\mathcal{X}_q$ with $\mathcal{X}_q= [g_{q-1},g_q]$, $g_{0} = x_L$ and $g_Q= x_R$. 
For simplicity we adopt a uniform mesh in which the number of collocation points keeps the same for all $q$, denoted by $M$, and use the Gauss-Lobbato points in each element. 
Then the spectral expansion for the coefficients $a_\nu(x,t)$ in $\mathcal{X}_q$ is
\begin{equation}
  \label{eq:approx_a_x}
  a_\nu(x,t)\approx \sum_{\mu=0}^{M-1}\beta_{\nu,\mu}(t)\phi_\mu(x),~~x\in\mathcal{X}_q,
\end{equation}
where 
\[
\phi_\mu(x)=\cos(\mu\theta),~~x = g_{q-1}+\frac{g_q-g_{q-1}}{2}(1-\eta),~~\eta = \cos(\theta) 
\]
with $\eta\in[-1,1]$ and $\theta\in[0,\pi]$. That is, $f(x,k,t)$ over $x\in\mathcal{X}_q$ and $k\in\mathcal{K}$ can be approximated by
\begin{equation}
  \label{eq:appro_f_kx}
  f(x,k,t) \approx \tilde{f}(x,k,t) = \sum_{\nu=-N/2+1}^{N/2}\sum_{\mu=0}^{M-1}\beta_{\nu,\mu}(t)\phi_\mu(x) \psi_\nu(k).
\end{equation}

Consequently, the partial derivative of $\tilde{f}(x,k,t)$ with respect to $x\in\mathcal{X}_q$ can be directly obtained as
\begin{equation}
  \label{eq:derivative_x}
  \nabla_x \tilde{f}(x,k,t) = \sum_{\nu=-N/2+1}^{N/2}\sum_{\mu=0}^{M-1}\tilde{\beta}_{\nu,\mu}(t)\phi_\mu(x) \psi_\nu(k)
\end{equation}
with
\begin{equation*}
  \tilde{\beta}_{\nu,\mu}(t)=-\frac{2}{g_q-g_{q-1}}\times\left\{
    \begin{array}{ll}
      0,& \mu= M-1,\\
      2(M-1)\beta_{\nu,M-1}(t),& \mu = M-2,\\
      \tilde{\beta}_{\nu,\mu+2}(t)+2(\mu+1)\beta_{\nu,\mu+1}(t), & \mu = M-3,\cdots,1,\\
      \frac{1}{2}\tilde{\beta}_{\nu,2}(t)+\beta_{\nu,1}(t),& \mu=0.
    \end{array} \right.
\end{equation*}
In a similar way, the partial derivatives of $\tilde{f}(x,k,t)$ with respect to $k\in\mathcal{K}$ 
have a much simpler expression: 
\begin{equation}\label{eq:approx_derivative}
  \nabla_k^{2l+1}\tilde{f}(x,k,t) = \sum_{\nu=-N/2+1}^{N/2}\sum_{\mu=0}^{M-1} (\frac{2\pi\mi \nu}{L_k})^{2l+1} \beta_{\nu,\mu}(t)\phi_\mu(x)\psi_{\nu}(k), ~~l=0,\ldots,[P/2].
\end{equation}

With the help of the {orthogonal relation of the Fourier basis:}
\begin{equation}
  \label{eq:integral_nmu}
  \begin{split}
    \int_{k_{\min}}^{k_{\max}}\me^{{2\pi}\mi(\nu+\zeta)k'/L_k}\D k' =  \left\{
      \begin{array}{ll}
        L_k, &    \nu+\zeta = 0,\\
         0, &  \nu+\zeta \neq 0,
      \end{array}  \right.
  \end{split}
\end{equation}
the truncated convolution term can be calculated analytically as follows
\begin{equation}
  \label{eq:truncate_g}
 \tilde{g}(x,k,t) := \int_{k_{\min}}^{k_{\max}}\tilde{V}_w^{loc}(x,k-k')\tilde{f}(x,k',t)\D k' = \sum_{\nu=-N/2+1}^{N/2}\sum_{\mu=0}^{M-1}b_{\nu,\mu}(t)\phi_\mu(x)\psi_{\nu}(k),
 \end{equation}
where the coefficients $b_{\nu,\mu}(t)$ are determined by $\beta_{\nu,\mu}(t)$:
\begin{equation}\label{eq:bnm}
b_{\nu,\mu}(t)=\beta_{\nu,\mu}(t) \frac{\left[V_{loc}(x-{y_\nu}/{2})-V_{loc}(x+{y_\nu}/{2})\right]}{\mi\hbar}.
\end{equation}

Finally, we obtain the following {semi-discretizated} scheme for the truncated Wigner equation \eqref{eq:truncated_W}
\begin{equation}\label{eq:spectral_complex}
  \begin{split}
    \frac{\partial}{\partial t}\tilde{f}(x,k,t)+\frac{\hbar k}{m}\nabla_x\tilde{f}(x,k,t) =\tilde{g}(x,k,t)+\sum_{l=0}^{[P/2]}\xi_l^{pol}(x)\nabla_k^{2l+1}\tilde{f}(x,k,t),
    \end{split}
\end{equation}
and the fast Fourier transform (FFT) can be used to accelerate the computation.

\subsection{Conservation laws}
\label{sec:conservation}

An explicit fourth-order Runge-Kutta method is used to evolve the 
{semi-discretizated} scheme \eqref{eq:spectral_complex} as we did in \cite{ShaoLuCai2011}.
Below we will show that the resulting full discretization scheme conserves both the mass and the energy. To this end, {it suffices} to consider the one-step forward Euler method with the time step $\Delta t$ and the resulting fully discretizated scheme is  
\begin{equation}\label{eq:spectral_complex_euler}
F^{n+1}(x,k) = F^{n}(x,k)+\Delta t \left[-\frac{\hbar k}{m}\nabla_xF^n(x,k)+G^n(x,k) 
       +\sum_{l=0}^{[P/2]}\xi_l^{pol}(x) \nabla_k^{2l+1}F^n(x,k)\right],
\end{equation}
where $F^n(x,k)$ and $G^n(x,k)$ denote the numerical solutions of $\tilde{f}(x,k,t)$ and $\tilde{g}(x,k,t)$ at time $t^n:=n\Delta t$, respectively.

To illustrate the numerical conservation laws, we need to consider the inner product $\langle \varphi, F^n \rangle$ in the computational domain $\Omega = \mathcal{X} \times \mathcal{K}$
\begin{equation}
\langle \varphi, F^n \rangle = \iint_{\mathcal{X}\times\mathcal{K}} \varphi(x, k) F^n (x, k) \D x \D k.
\end{equation}
and the numerical current density 
\begin{equation}
j^n(x) = \int_{\mathbb{\mathcal{K}}} k F^n(x, k) \D k.
\end{equation}

\begin{proposition} \label{pro:mass_con}
The numerical scheme \eqref{eq:spectral_complex_euler} conserves the mass, 
i.e.,  
\begin{equation}
\langle 1, F^{n+1} \rangle = \langle 1, F^n \rangle
\end{equation}
provided that the total inflow and outflow are in balance, say, 
\begin{equation}\label{balance_condition}
j^n(x_L) = j^n(x_R).
\end{equation}
\end{proposition}
\begin{proof}
Through integration by parts, it is easy to verify 
\[
\langle 1, -\frac{\hbar k}{m} \nabla_x F^n \rangle = 0 \quad \textup{and} \quad   \langle 1, \xi_l^{pol} \nabla_k^{2l+1}F^n \rangle = 0
\]
due to the Eq.~\eqref{balance_condition} as well as the periodic condition in $k$-space
\[
\nabla_{k}^{2l} F^n(x, k+L_k) = \nabla_{k}^{2l} F^n(x, k), \quad l=0,1,\cdots, [P/2].
\]
So we only need to show $\langle 1, G^n\rangle = 0$. 

Splitting the summation with respect to $\nu$ into two parts, one for $\nu\neq 0$ and the other for $\nu=0$, {it} leads to
\begin{equation}\label{eq:1G}
\langle 1, G^n \rangle =  \sum_{\nu\neq 0} (\int_{\mathcal{K}}\psi_{\nu}(k) \D k) (\sum_{\mu=0}^{M-1}b_{\nu,\mu}(t) \int_{\mathcal{X}} \phi_\mu(x) \D x) + L_k \sum_{\mu=0}^{M-1}b_{0,\mu}(t)  \int_{\mathcal{X}}\phi_\mu(x) \D x  = 0,
\end{equation}
where we have applied in order Eq.~\eqref{eq:integral_nmu} and the fact that $b_{0,\mu}(t) \equiv 0$ for any $\mu\in\{0,1,\ldots,M-1\}$ according to Eq.~\eqref{eq:bnm}.
\end{proof}

The numerical energy conservation, 
however, requires some additional conditions, 
as stated below.

\begin{proposition} \label{pro:energy_con}
The numerical scheme \eqref{eq:spectral_complex_euler} conserves the energy, 
i.e.,  
\[
\langle H, F^{n+1} \rangle = \langle H, F^n \rangle, 
\]
where $H(x,k) = \frac{\hbar^2k^2}{2m} +V_{pol}(x) + V_{loc}(x)$,
provided that 

(a) $V_{loc}(x)\in L(\mathbb{R}) \cap C^\omega(\mathbb{R})$;

(b) $\forall\, k \in \mathcal{K}$, $F^n(x_L, k) = F^n(x_R, k) = 0$; 

(c) $\forall\, x \in \mathcal{X}$, $\nabla_{k}^{l} F^n(x, k_{\min}) = \nabla_{k}^{l} F^n(x, k_{\max}) = 0, ~~ l=0,1, \cdots, +\infty$.

\end{proposition}

\begin{proof}

We intend to prove the following relations: 
\begin{equation}\label{eq:relation}
\left\{
\begin{split}
&\langle H, -\frac{\hbar k}{m} \nabla_x F^n + \xi_{0}^{pol} \nabla_{k} F^n + \xi_{0}^{loc} \nabla_{k} F^n \rangle = 0, \\
&\langle H, G_n - \xi_{0}^{loc} \nabla_{k} F^n \rangle = 0, \\
&\langle H,  {\sum_{l=1}^{[P/2]}\xi^{pol}_l} \nabla_k^{2l+1} F^{n} \rangle = 0.
\end{split}
\right.
\end{equation}

Using the integration by parts leads to  
\begin{align*}
 \langle H, - \nabla_k H \nabla_x F^n  \rangle &= \langle \nabla_k H \nabla_x H, F^n \rangle, \\ 
 \langle H,  \nabla_x H \nabla_k F^n \rangle &= - \langle \nabla_x H \nabla_k H, F^n \rangle,
\end{align*}
where both conditions (\textup{b}) and (\textup{c}) are applied to eliminate the boundary terms, 
and then we arrive at the Liouville theorem 
\[
\langle H, -\nabla_k H \nabla_x F^n + \nabla_x H \nabla_k F^n \rangle = 0,
\]
which is nothing but the first relation of Eq.~\eqref{eq:relation}.

For $V_{loc}(x) \in C^\omega(\mathbb{R})$, by the Taylor theorem, we have 
\begin{align*}
 V_{loc}(x+{y_\nu}/{2})-V_{loc}(x-{y_\nu}/{2}) &= \sum_{l=0}^{+\infty}\frac{y_\nu^{2l+1}}{2^{2l}}\frac{\nabla_x^{2l+1} V_{loc}(x)}{(2l+1)!}, \\
 \frac{V_{loc}(x+{y_\nu}/{2})-V_{loc}(x-{y_\nu}/{2})}{y_\nu}-\nabla_xV_{loc}(x) &=\sum_{l=1}^{+\infty}\left(\frac{2\pi\mi\nu}{L_k}\right)^{2l+1}\frac{(-1)^l}{2^{2l}}\frac{\nabla_x^{2l+1} V_{loc}(x)}{(2l+1)!}.
\end{align*}
Substituting it into $G^n(x, k)$ yields 
\begin{equation*}
 G^n(x,k)-\xi^{loc}_0(x)\nabla_kF^n(x,k) = \sum_{l=1}^{+\infty}\xi^{loc}_l(x)\cdot\nabla_k^{2l+1}F^n(x,k), \end{equation*}
implying that, for the remaining two relations of Eq.~\eqref{eq:relation},  
we only need to verify 
\[
\langle H,\xi^{loc}_l \nabla_k^{2l+1}F^n\rangle = 0 \,\,\,\text{for}\,\, l = 1,2,\ldots, +\infty,
\]
and 
\[
\langle H, \xi^{pol}_l \nabla_k^{2l+1} F^n \rangle  = 0 \,\,\, \text{for}\,\, l = 1, 2, \ldots, [P/2],
\]
both of which must vanish through the integration by parts due to $\nabla_k^{2l+1}H \equiv 0$ for $l\geq 1$ and the condition \textup{(c)}. 

\end{proof}

\subsection{Splitting treatment}
\label{sec:splitting}

Finally, we would like to mention that the splitting treatments of unbounded potentials used in Eq.~\eqref{eq:key1} are very useful when some resulting subproblems allow exact solutions. Moreover, within the framework of the splitting schemes,  different methods could be used to tackle $V_{loc}(x)$ and $V_{pol}(x)$ separately. For example, we may solve the truncated Wigner equation \eqref{eq:truncated_W} in a splitting manner
\begin{equation}\label{eq:split_1}
\begin{cases}
\displaystyle
\text{(A)}\quad \frac{\partial}{\partial t}f(x,k,t)+\frac{\hbar k}{m}\nabla_xf(x,k,t) = \int_{k_{\min}}^{k_{\max}}\tilde{V}_w^{loc}(x,k-k')f(x,k',t)\D k', \\
\displaystyle
\text{(B)}\quad \frac{\partial}{\partial t}f(x,k,t) = \sum_{l=0}^{[P/2]}\xi_l^{pol}(x)\nabla_k^{2l+1}f(x,k,t),
\end{cases}
\end{equation}
where the subproblem (B) has explicit solutions for the quadratic potential $V_{pol}(x) = x^2-1$ \cite{SellierDimov2015-HO}.

Next, we consider the conservation laws of the splitting methods. It only needs to consider the simplest Lie-Trotter scheme, with the same spectral collocation method adopted for both subproblems,
namely, 
\begin{equation}\label{eq:split_num1}
\begin{cases}
\displaystyle
 F^{n+\frac{1}{2}}(x,k) = F^{n}(x,k)+\Delta t \left[-\frac{\hbar k}{m}\nabla_x F^n(x,k)+G^n(x,k)\right],\\
 \displaystyle
F^{n+1}(x,k) = F^{n+\frac{1}{2}}(x,k) + \Delta t\sum_{l=0}^{[P/2]}\xi_l^{pol}(x) \nabla_k^{2l+1}F^{n+\frac{1}{2}}(x,k),
\end{cases}
\end{equation}
and thus we arrive at 
\begin{align}
    F^{n+1}(x,k) = F^{n}(x,k)&+\Delta t \left[-\frac{\hbar k}{m}\nabla_xF^n(x,k)+G^n(x,k)+\sum_{l=0}^{[P/2]}\xi_l^{pol}(x) \nabla_k^{2l+1}F^n(x,k) \right] \nonumber \\
    &+ \Delta t^2\sum_{l=0}^{[P/2]}\xi_l^{pol}(x) \nabla_k^{2l+1}  \left[-\frac{\hbar k}{m} \nabla_xF^n(x,k)+G^n(x,k) \right].   \label{eq:numerical_split_1}
\end{align}

\begin{proposition} \label{pro:mass_con_split}
The splitting scheme \eqref{eq:numerical_split_1} conserves the mass, 
i.e.,  
\[
\langle 1, F^{n+1} \rangle = \langle 1, F^n \rangle ,
\]
provided that 

(a) $j^n(x_L) = j^n (x_R)$;

(b) $\forall\, x \in \mathcal{X}$, $\nabla_{k}^{2l} F^n(x, k_{\min}) = \nabla_{k}^{2l} F^n(x, k_{\max}) = 0$ for $l=0, \ldots, [P/2]$.
\end{proposition}
\begin{proof}
Comparing the scheme \eqref{eq:numerical_split_1} with \eqref{eq:spectral_complex_euler} and according to Proposition \ref{pro:mass_con}, we only need to show  
\[
\sum_{l=0}^{[P/2]} \langle 1,   \xi_l^{pol} \nabla_k^{2l+1}G^n \rangle  + \langle 1, \xi_l^{pol}  \nabla_k^{2l+1} (-\frac{\hbar k}{m} \nabla_x F^n ) \rangle = 0.
\]
Actually, each term in the left-hand-side of the above equation must vanish.  

As for the first term, it is a direct outcome of $\langle 1,   \nabla_k^{2l+1}G^n \rangle=0$ which can be readily verified through the integration by parts as well as using the
periodicity of $G^n$ in $k$-space. 

From condition (b), a direct calculation shows 
\begin{align*}
\int_{\mathcal{K}} \nabla_{k}^{2l+1} (-\frac{\hbar k}{m}F^n(x, k) ) \D k =& \left.\nabla_k^{2l} (-\frac{\hbar k}{m} F^n(x, k) \right|_{k_{\min}}^{k_{\max}}  \\
 =& \left. (- \frac{2l \hbar}{m} \nabla^{2l-1}_k F^n(x, k)  - \frac{\hbar k}{m} \nabla_k^{2l} F^n(x, k))\right|_{k_{\min}}^{k_{\max}} = 0, ~\forall\, x \in \mathcal{X}, 
\end{align*}
and then performing the integration by parts in $x$-space for the second term yields  
\begin{align*}
\langle 1, \xi_l^{pol} \nabla_k^{2l+1} (-\frac{\hbar k}{m} \nabla_x F^n ) \rangle 
= &  \left. \{\xi_l^{pol}(x)[\int_{\mathcal{K}} \nabla_k^{2l+1}  (-\frac{\hbar k}{m}F^n(x, k)) \D k]\} \right|_{x_L}^{x_R} \\
&-\int_{\mathcal{X}}  \nabla_x\xi_l^{pol}(x) [\int_{\mathcal{K}} \nabla_k^{2l+1}  (-\frac{\hbar k}{m}F^n(x, k)) \D k] \D x =0.
\end{align*}
The proof is finished. 
\end{proof}

Unfortunately, the numerical energy conservation fails to hold for the splitting treatment because
\begin{equation}\label{eq:integral_con}
\langle H,  \sum_{l=0}^{[P/2]}\xi_l^{pol} \nabla_k^{2l+1}(-\frac{\hbar k}{m}\nabla_x F^n+G^n) \rangle = 0
\end{equation}
cannot be guaranteed for the scheme \eqref{eq:numerical_split_1}.
This can be readily verified by taking, for instance,  $V_{pol}(x)=x^2-1$ (see more details in Section \ref{sec:rational}).

\section{Numerical experiments}
\label{sec:4}

In this section, several typical quantum systems are employed to test the performance of the {proposed} methods and the atomic units $\hbar =m = e= 1$ are used if not specified.
We employ the $L^2$-error $\varepsilon_2(t)$ and the $L^{\infty}$-error $\varepsilon_{\infty}(t)$ to study the convergence rate of the spectral collocation method:
\begin{align}
  \label{eq:err_2}
  \varepsilon_2(t)& = \left[{\iint_{\mathcal{X}\times\mathcal{K}}} (f^{\text{num}}(x,k,t)-f^{\text{ref}}(x,k,t))^2\D x\D k \right]^{1/2},\\
  \label{eq:err_infty}
  \varepsilon_\infty(t) &= \max_{(x,k)\in \Omega}\{|f^{\text{num}}(x,k,t)-f^{\text{ref}}(x,k,t)|\},
\end{align}
where $f^{\text{num}}(x,k,t)$ is the numerical solution, and $f^{\text{ref}}(x,k,t)$ the reference solution which could be {either} the exact solution or the numerical solution on {the finest} grid mesh. 
In order to monitor the numerical conservation of mass and energy, the variations of total mass $\varepsilon_{\text{mass}}(t)$ and energy $\varepsilon_{\text{energy}}(t)$ are also chosen as the metrics, 
\begin{align}
  \label{eq:epsilon_mass}
  \varepsilon_{\text{mass}}(t) &={\iint_{\mathcal{X}\times\mathcal{K}} (f^{\text{num}}(x,k,t) - f^{\text{num}}(x,k,0))\D x\D k},\\
  \label{eq:epsilon_energy}
  \varepsilon_{\text{energy}}(t) &= {\iint_{\mathcal{X}\times\mathcal{K}} H(x,k)(f^{\text{num}}(x,k,t)- f^{\text{num}}(x,k,0))\D x\D k.}
\end{align}
As we did in \cite{ShaoLuCai2011,XiongChenShao2016}, all above metrics are evaluated by a simple rectangular rule over a uniform mesh.

\subsection{The P\"oschl-Teller potential}
\label{sec:pt}

The P\"oschl-Teller potential
\begin{equation}
  \label{eq:pt}
  V(x) = -\lambda(\lambda+1)\frac{\hbar^2}{2m}\text{sech}^2(x), \quad \lambda\in\mathbb{N}, 
\end{equation}
serves as the first example 
for its energy levels $E_{\lambda,n}$ and bound states $\phi_{\lambda,n}(x)$ with $n=0,1,\ldots,\lambda-1$ can be obtained analytically \cite{bk:Flugge2012, BundTijero2000}.  

For instance, when $\lambda=1,2$,  we have
\begin{align*}
  \lambda=1:~&E_{1,0} = -\frac{\hbar^2}{2m}, \; \phi_{1,0}(x) = \frac{\sqrt{2}}{2}\text{sech}(x);\\
  \lambda=2:~&E_{2,0} = -\frac{2\hbar^2}{m},\;\phi_{2,0}(x) = \frac{\sqrt{3}}{2}\text{sech}^2(x), \\
  &E_{2,1} = -\frac{\hbar^2}{2m}, \; \phi_{2,1}(x) = \frac{\sqrt{6}}{2}\text{sech}^2(x)\sinh(x).
\end{align*}

\begin{figure}
\includegraphics[width=0.5\textwidth,height=0.35\textwidth]{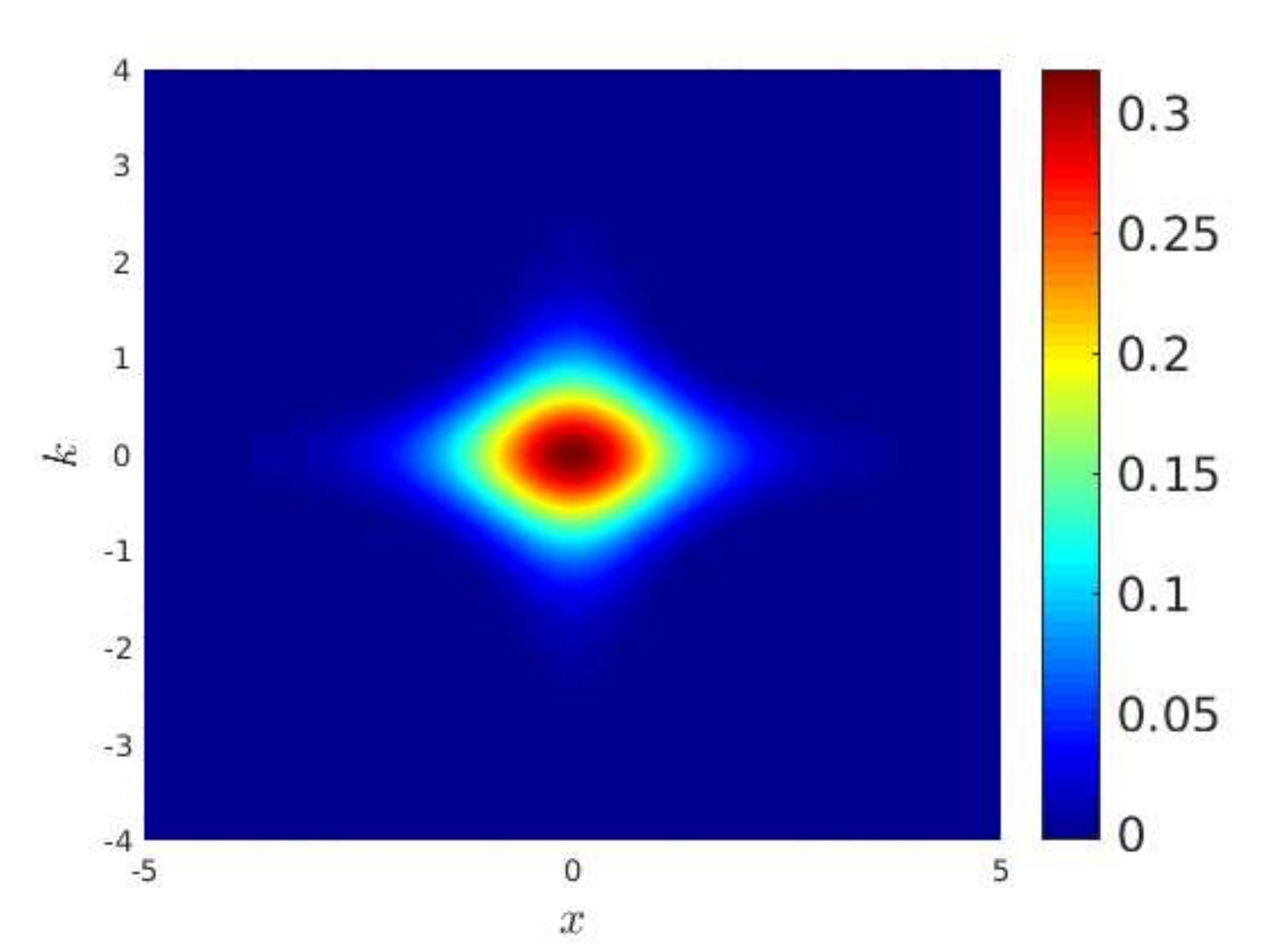}
\includegraphics[width=0.5\textwidth,height=0.35\textwidth]{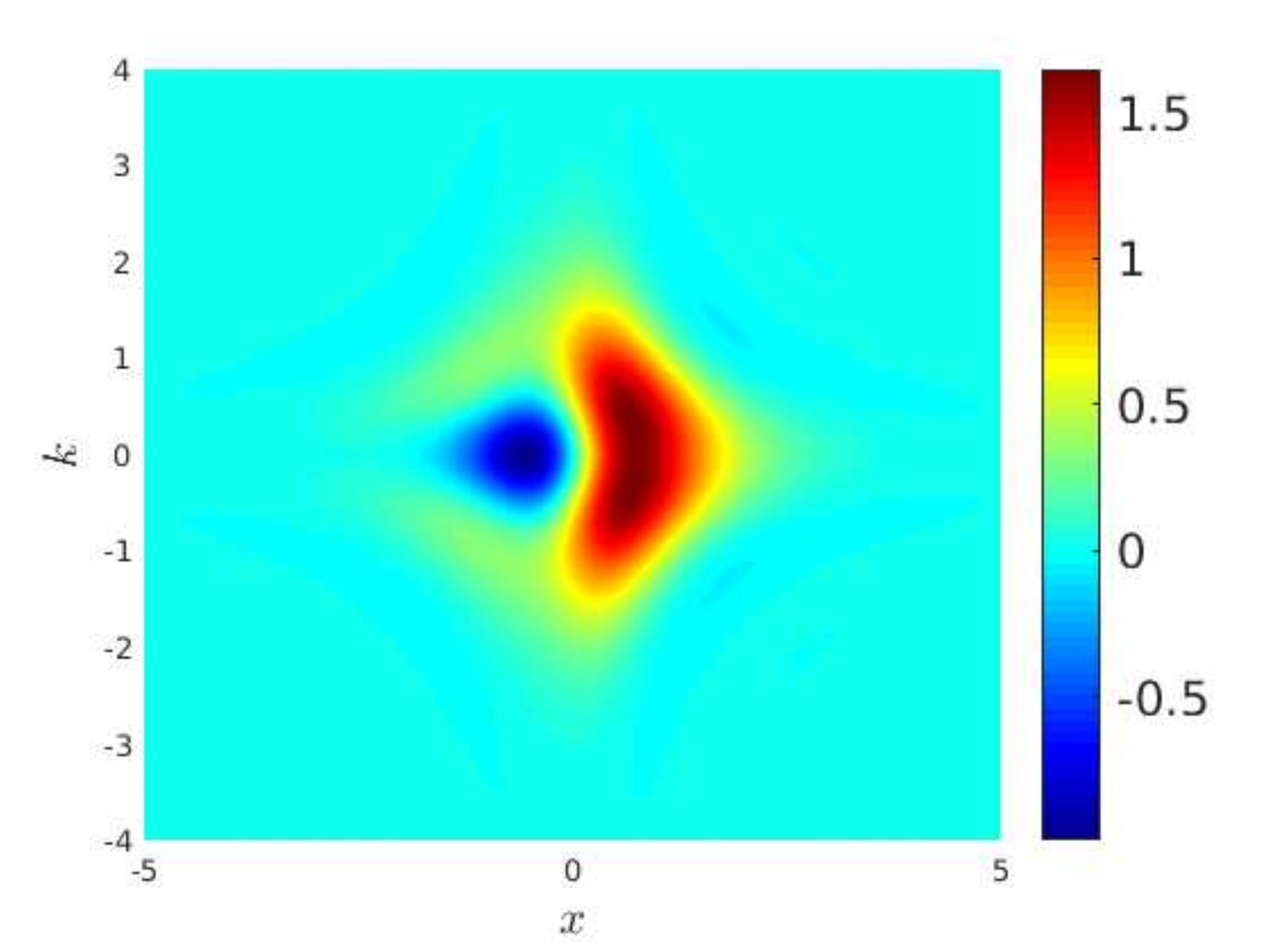}
\caption{\small The Wigner function under the P\"oschl-Teller potential. Left: The stationary state $f_{1, 0}(x,k)$ in Eq.~\eqref{eq:stationary_1}; Right: The superposed state $f_{2,1,0}(x, k, 0)$ in Eq.~\eqref{eq:initial_2}.}
\label{fig:pt2}
\end{figure}

The Wigner function corresponding to $\phi_{1,0}$, as shown in the left plot of Fig.~\ref{fig:pt2}, reads 
\begin{equation}
  \label{eq:stationary_1}
    f_{1, 0}(x,k) = \int_{-\infty}^{+\infty}\phi_{1,0}\left(x+\frac{y}{2}\right)\phi^\ast_{1,0}\left(x-\frac{y}{2}\right)\me^{-\mi ky}\D y = \frac{\sin(2xk)}{\hbar\sinh(2x)\sinh(\pi k)}.
\end{equation}

\begin{figure}
  \includegraphics[width=0.5\textwidth,height=0.35\textwidth]{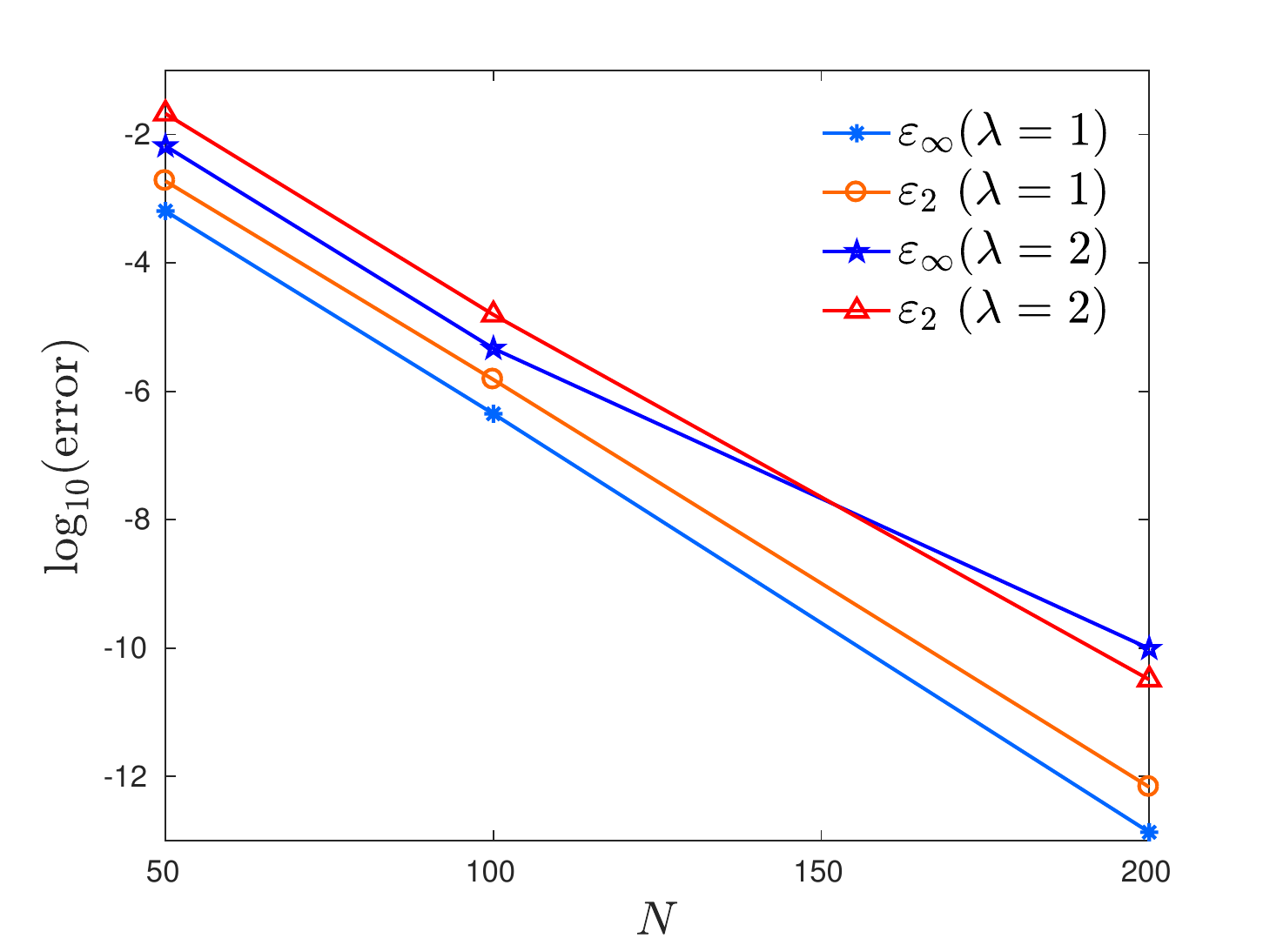}
  \includegraphics[width=0.5\textwidth,height=0.35\textwidth]{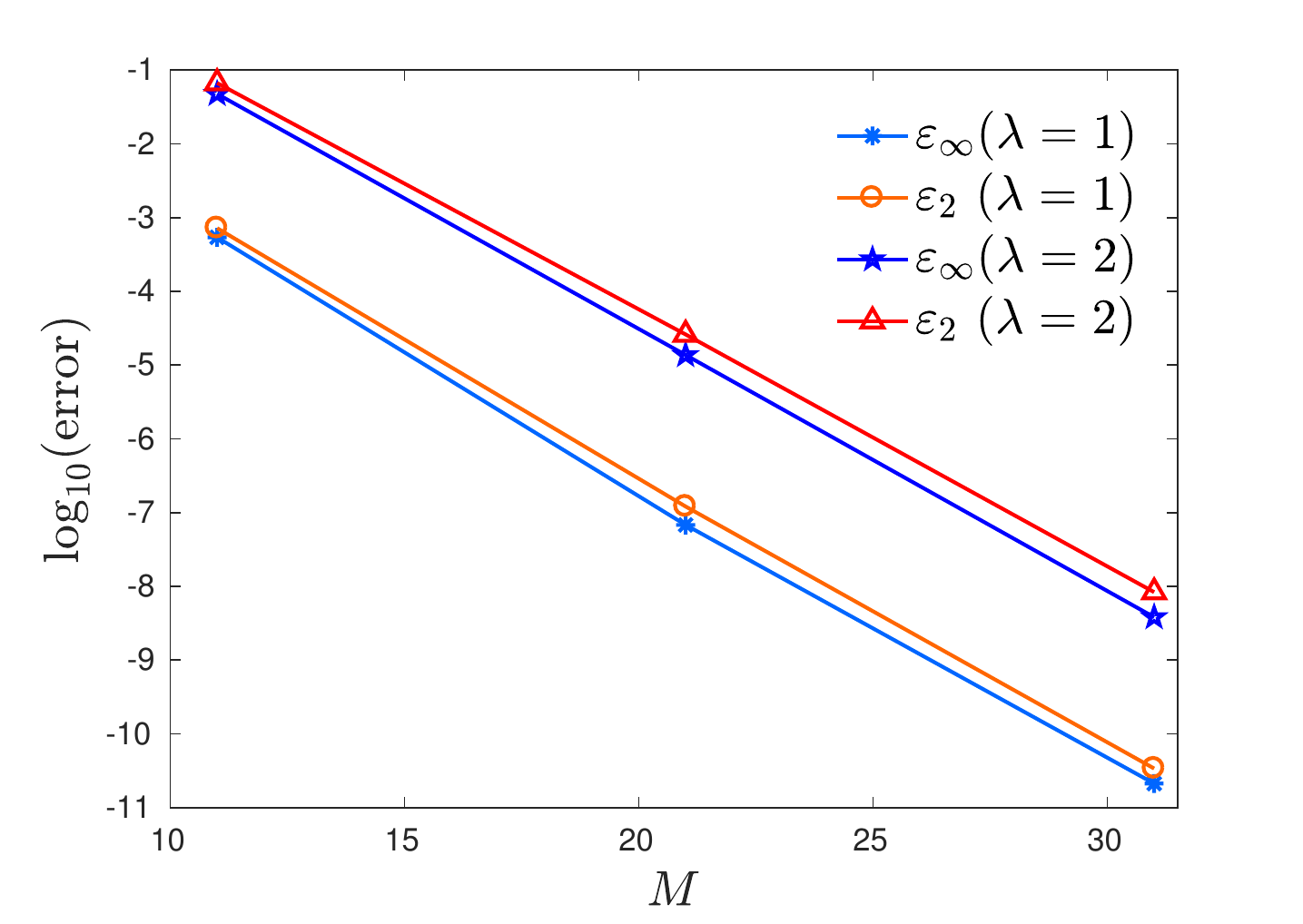}
  \caption{\small The P\"oschl-Teller potential: The convergence rate with respect to $N$ (left) and $M$ (right). The spectral convergence in both $x$-space and $k$-space is observed for $\lambda = 1,2$. All the errors are measured at the instant $t=10$.}
   \label{fig:converg_pt}
\end{figure}

The Wigner function corresponding to the superposed state 
$\frac{\sqrt{2}}{2}\phi_{2,0}(x)\me^{-\mi E_{2,0}t/\hbar}+\frac{\sqrt{2}}{2}\phi_{2,1}(x)\me^{-\mi E_{2,1}t/\hbar}$  is 
\begin{equation}
  \label{eq:initial_2}
  \begin{split}
  f_{2,1,0}(x,k,t) = &\frac{3}{8}\int_{-\infty}^{+\infty} \text{sech}^2(x+y/2)\text{sech}^2(x+y/2)\times [2\sinh (x+y/2)\sinh(x-y/2)\\
  & + \sqrt{2}\sinh(x-y/2)\me^{\mi 3\hbar^2t/2m}+\sqrt{2}\sinh(x+y/2)\me^{-\mi 3\hbar^2t/2m} + 1] \me^{-\mi k y} \D y,
  \end{split}
\end{equation}
and the energy difference there is 
\begin{equation}\label{eq:dE}
\Delta E=E_{2,1}-E_{2,0}=\frac{3\hbar^2}{2m} = 1.5.
\end{equation}
We show $f_{2,1,0}(x, k, 0)$ in the right plot of Fig.~\ref{fig:pt2}.

Two groups of simulations are performed: $f_{1, 0}(x,k)$ in Eq.~\eqref{eq:stationary_1} is used as the initial data in the first while $f_{2,1,0}(x, k, 0)$ in Eq.~\eqref{eq:initial_2} in the second. Other parameters are chosen as:
$-x_L=x_R=20$, $-k_{\min}=k_{\max}=\frac{10\pi}{3}$, $Q=10$ and $\Delta t=0.0005$.
To study the convergence rate with respect to $N$, the number of collocation points in each $x$-element is fixed to be $M = 41$. Similarly, when studying the convergence rate with respect to $M$, the number of collocation points in $k$-space is fixed to be $N = 256$. As shown in Fig.~\ref{fig:converg_pt}, the spectral convergence with respect to both $N$ and $M$ can be clearly observed. 
When $\lambda=1$, the numerical Wigner function is found to be almost at rest,
and the numerical errors at $t=10$ are no more than $10^{-13}$ on the finest mesh $(N,M) = (256,41)$. When $\lambda = 2$, as predicted by Eq.~\eqref{eq:initial_2}, 
the Wigner function rotates around the center periodically with the period $1.5$, referring to the energy level transition under the P\"{o}schl-Teller potential. Nevertheless, the numerical errors are still no more than $10^{-10}$ until $t=10$ on the finest mesh. The left plot of Fig.~\ref{fig:mean_pt2} shows the averaged displacement $\langle x\rangle$ and the averaged momentum $\langle k\rangle$ up to $t=50$ and a simple periodic mode is so evident. A direct spectrum analysis gives us a frequency of $1.508$, see the right plot of Fig.~\ref{fig:mean_pt2}, which accords with the theoretical value $\Delta E$ in Eq.~\eqref{eq:dE}.

In order to verify the numerical conservation laws, we record $\varepsilon_{\text{mass}}(t)$ and $\varepsilon_{\text{energy}}(t)$ during the simulations and find out: When $\lambda=1$, $\varepsilon_{\text{mass}}(t)$ is no more than $1.4766\times 10^{-14}$ and $\varepsilon_{\text{energy}}(t)$ is no more than $6.4893\times 10^{-13}$ until $t=10$, 
which is comparable to the errors on the boundary (around $10^{-14}$); If we enlarge the computational domain to be $-x_L=x_R=100$, $-k_{\min}=k_{\max}={100\pi}$ to guarantee the Wigner function vanishes outside the computational domain, then both $\varepsilon_{\text{mass}}(t)$ and $\varepsilon_{\text{energy}}(t)$ are around the machine epsilon until $t=10$ even on a very coarse mesh, say, $Q = 1$, $M=11$ and $N=4$. That is, our proposed spectral discretization is indeed mass-and-energy-conserving as predicted by Propositions \ref{pro:mass_con} and \ref{pro:energy_con}.

\begin{figure}
   \centering
   \includegraphics[width=0.49\textwidth,height=0.35\textwidth]{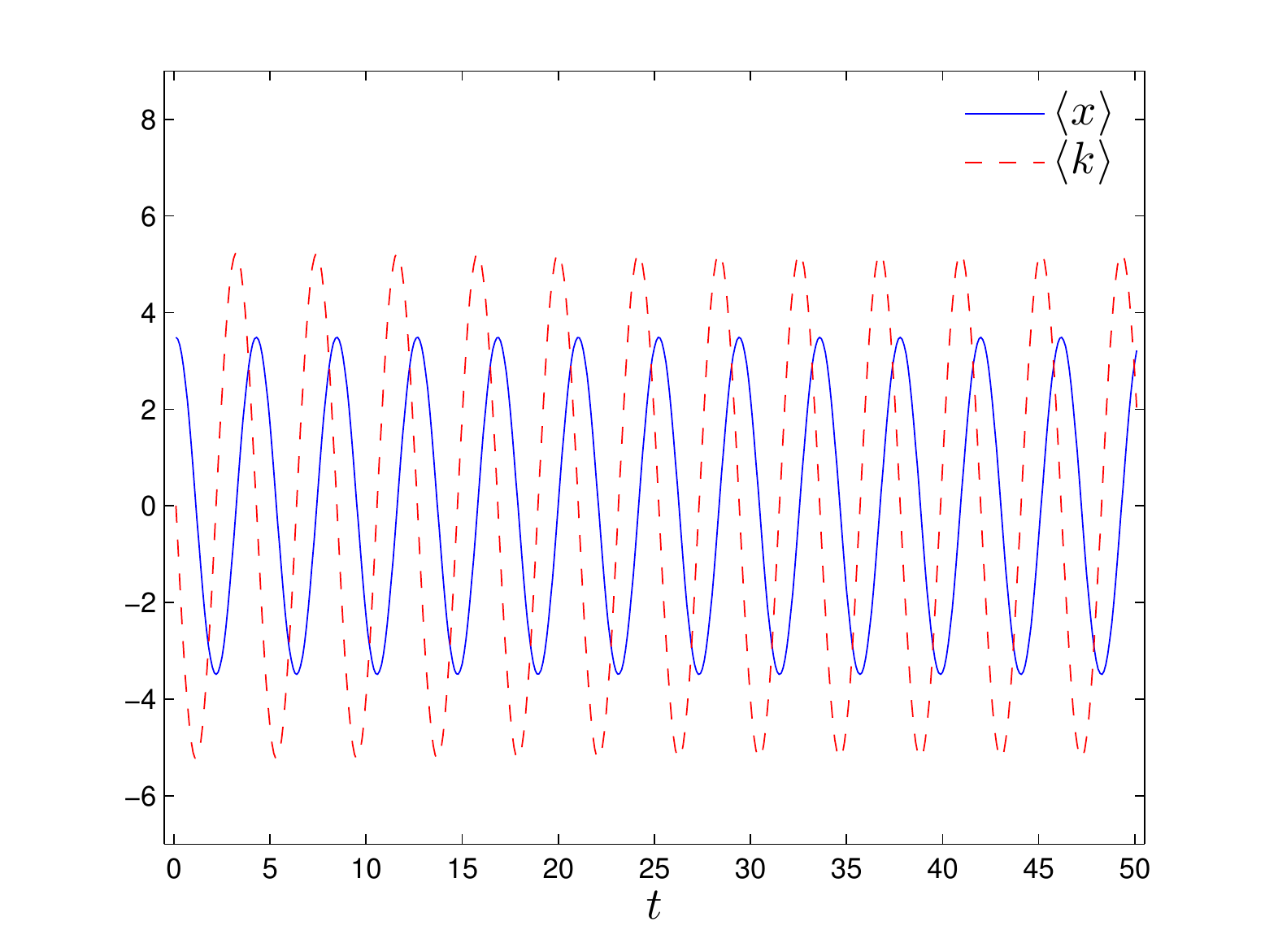}
   \includegraphics[width=0.49\textwidth,height=0.35\textwidth]{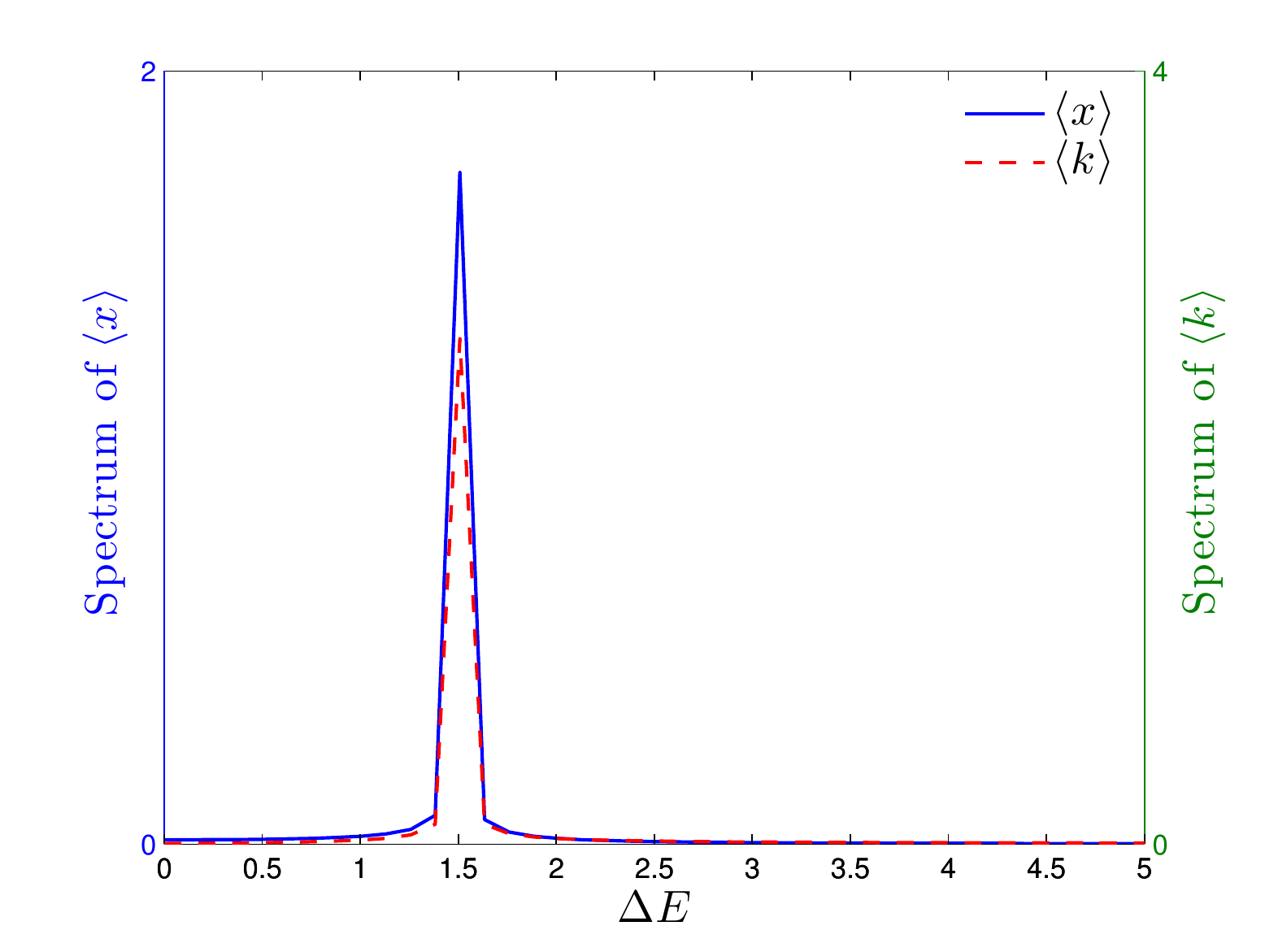}
   \caption{\small The P\"oschl-Teller potential with $\lambda=2$. Left: The averaged displacement $\langle x\rangle$ and the averaged momentum $\langle k\rangle$ along the time. Right: The spectrum analysis gives a frequency of $1.508$ which accords with the theoretical value $\Delta E$ in Eq.~\eqref{eq:dE}.}
    \label{fig:mean_pt2}
\end{figure}

\subsection{Fourth-order anharmonic oscillators}
\label{sec:anharmonic}

We turn to discuss a class of very simple, but rather important unbound potentials, termed {fourth-order} anharmonic oscillators, and eight fourth-order double-well potentials are chosen from \cite{SomorjaiHornig1962}, the parameters of which are presented in Table~\ref{tab:4apot}. 
The symmetric potentials $V_1 \sim V_4$ there are listed in order of decreasing barrier height $h$,
while the asymmetric ones $V_5\sim V_{8}$ in order of increasing gap $g$.

\begin{table}
\centering
\caption{\small Eight fourth-order double-well potentials adopted from \cite{SomorjaiHornig1962}. 
The potential is characterized by $V(x) = \frac12 (v_2x^2+v_3x^3+v_4x^4)$
and three parameters $v_2, v_3, v_4$ are chosen to adjust its two local {minima} $g-h$, $-h$ as well as the width $w$, see the left plot. The symmetric potentials $V_1 \sim V_4$ are listed in order of decreasing barrier height $h$ while the asymmetric ones $V_5\sim V_{8}$ in order of increasing gap $g$.}
{\scriptsize{
  \begin{tabular}{cccccccccc}
  \hline
   \hline
           \multirow{7}{*}{ 
           \includegraphics[scale=0.22]{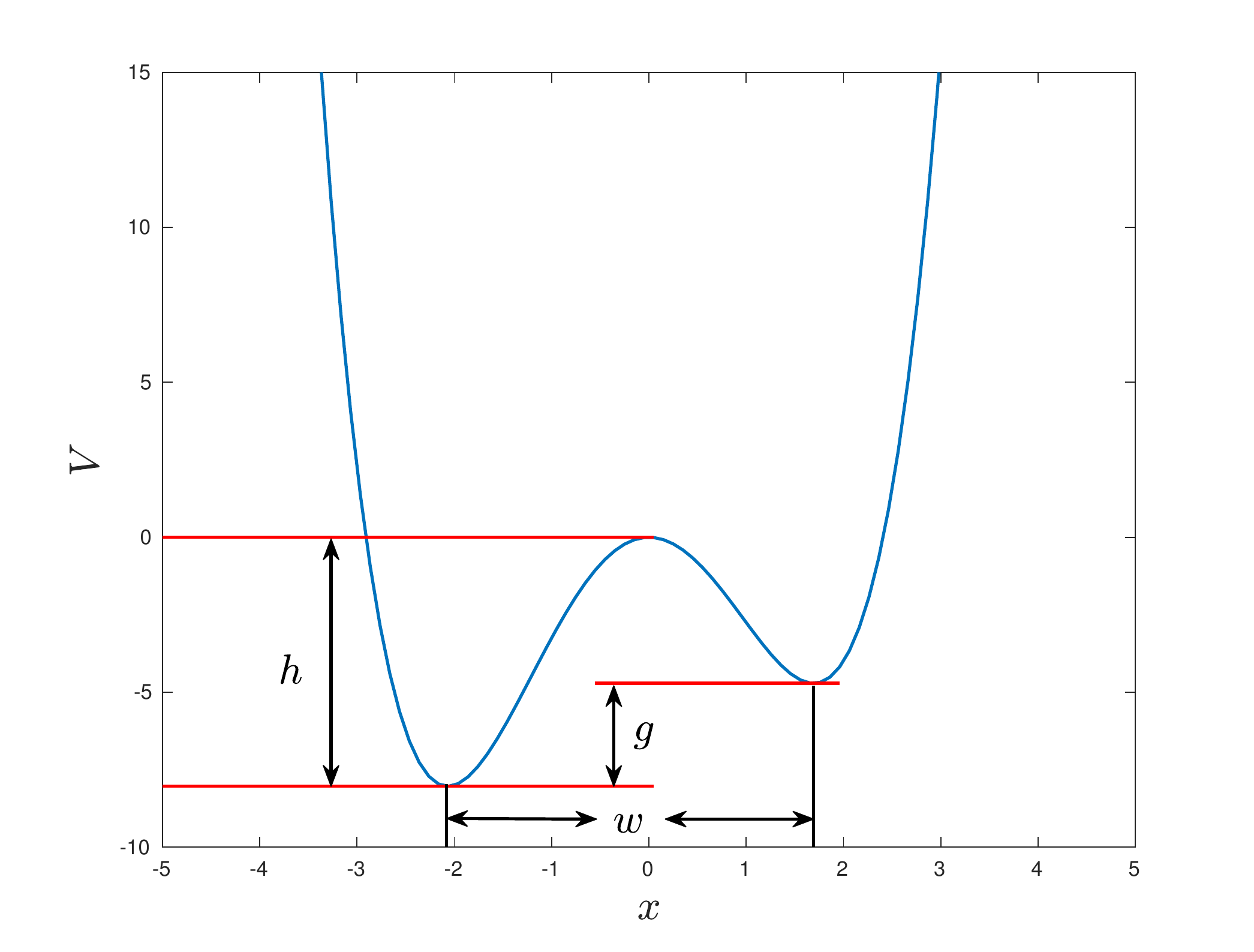}} &  &  $V_1$  &   $V_2$  &   $V_3$  &  $V_4$  &  $V_5 $
              &  $V_6$  &   $V_7$  &   $V_8$    \\  
    \cline{2-10}
    & $v_4$    &  0.6575 &   0.32   &  0.20   &  0.05   &  1.00
              &   0.98   &  1.00   &  1.00      \\ 
     \cline{2-10}
    & $v_3$     &    0    &    0     &    0    &    0    
              &  0.50   &  0.4939  & 0.5812  &  3.00     \\
     \cline{2-10}
   &  $v_2$     &  -5.26  &  -2.56   &  -1.60  &  -0.40  
              &  -7.00  &  -7.77   & -7.9050 &  -1.00    \\
\cline{2-10}
& $w$     &  4.00   &   4.00   &  4.00   &  4.00   
              & 3.7625  &   4.00   &  4.00   & 2.6575     \\
   \cline{2-10}
        & $h$     &   5.26  &   2.56   &  1.60   &   0.40  
              & 8.0465  &  9.9553  & 10.5165 & 7.0455     \\
              \cline{2-10}
              & $g$     &    0    &     0    &    0    &    0    
              & 3.3235  &  3.9510  & 4.6495  & 7.0433     \\
    \hline
    \hline
  \end{tabular}
 }}
 \label{tab:4apot}
\end{table}

We will show that the energy level transitions can be accurately captured by
 our conservative spectral solver for the Wigner quantum dynamics. 
  To this end, a Gaussian wave packet of the form:
 \begin{equation}\label{init_GWP}
 f_0(x,k,0)=A\exp(-\sigma_1(x-x_0)^2-\sigma_2(k-k_0)),
 \end{equation}
 is set to be the initial state where $A$ is the normalizing constant. 
 In the first group of simulations, we set $A =2/\pi$, $\sigma_1 = 4$, $\sigma_2 = 1$,  $x_0=0$, $k_0=0.5$ for the initial data, the time step $\Delta t = 10^{-5}$
 and the mesh size to be $N = 200$, $Q = 10$ and $M = 21$. The computational domain is $-x_L=x_R=15$, $-k_{\min}=k_{\max}={10\pi}/{3}$ for $V_1\sim V_4$ and $-x_L=x_R=10$, $-k_{\min}=k_{\max}={5\pi}$ for $V_5 \sim V_8$.

\begin{table}[h]
  \centering
  \caption{\small Energy level transitions in a fourth-order anharmonic oscillator with the symmetric potential $V_4$. $\Delta E_{mn}$ denotes the energy difference for the transition between the $n$-th and the $m$-th levels. The reference (Ref.) value of $\Delta E_{mn}$ is calculated by $\Delta E_{mn} = |E_m-E_n|$, whereas the numerical (Num.) value is directly obtained by the spectral analysis of either the averaged displacement $\langle x\rangle$ or momentum $\langle k\rangle$ from the numerical evolution of the Wigner equation until $t=100$ (left) and $t=1000$ (right) by the conservative spectral method, see Fig.~\ref{fig-v4}.
  The nine lowest energy levels $E_n$ with $n = 0,1,\ldots, 8$ for the reference are obtained by the highly accurate Pruess method with the relative error tolerance of $10^{-10}$\cite{PruessFulton1993,ShaoCaiTang2006}.}
 \begin{tabular}{cccc|cccc}
   \hline
   \hline
        \multicolumn{4}{c|}{$t=100$}  &  \multicolumn{4}{c}{$t=1000$}\\
   \hline
     $\Delta E_{mn}$ & Ref. &   {Num.}   &     Error      &
     $\Delta E_{mn}$ & Ref. &   {Num.}   &     Error      \\
   \hline
       $E_1-E_0$ &    0.1845    &    0.1885    &     0.0040     &
       $E_1-E_0$ &    0.1845    &    0.1822    &     0.0023     \\
   \hline
       $E_2-E_1$ &    0.5465    &  {\it 0.5655}  &   0.0190     &
       $E_2-E_1$ &    0.5465    &    0.5466    &     0.0001     \\
   \hline
       $E_3-E_2$ &    0.5897    &  {\it 0.5655}  &   0.0242     &
       $E_3-E_2$ &    0.5897    &    0.5906    &     0.0009     \\
   \hline
                &              &               &              &          
       $E_4-E_3$ &    0.6977    &    0.6974    &     0.0003     \\
   \hline
       $E_5-E_4$ &    0.7738    &  {\it 0.8168}  &   0.0430      &
       $E_5-E_4$ &    0.7738    &    0.7728    &     0.0010     \\
   \hline
       $E_6-E_5$ &    0.8401    &  {\it 0.8168}  &   0.0233     &
       $E_6-E_5$ &    0.8401    &    0.8419    &     0.0018     \\       
   \hline
                 &              &              &             &
       $E_7-E_6$ &    0.8984    &    0.8985    &     0.0001     \\
   \hline
       $E_8-E_7$ &    0.9508    &    0.9425    &     0.0083     &
       $E_8-E_7$ &    0.9508    &    0.9488    &     0.0020     \\
   \hline
       $E_3-E_0$ &    1.3207    &    1.3195    &     0.0012     &
       $E_3-E_0$ &    1.3207    &    1.3195    &     0.0012     \\
   \hline
       $E_5-E_2$ &    2.0612    &    2.0735    &     0.0123     &
       $E_5-E_2$ &    2.0612    &    2.0609    &     0.0003     \\
   \hline
   \hline                                                           
  \end{tabular}
   \label{tab:v4_energy}
\end{table}

Taking the symmetric potential $V_4$ as an example, we record the signals of $\langle x\rangle$ and $\langle k\rangle$ until $t=100$ and analyze their spectrum, as demonstrated in left plot of Fig.~\ref{fig-v4},
and the main energy level transitions corresponding to the peaks of spectrum are collected in left part of Table~\ref{tab:v4_energy}.
We can observe there that the peaks of spectrum for $\langle x\rangle$ accord perfectly with those for $\langle k\rangle$, 
and the location of each peak, denoted by $\Delta E$, gives the energy transition between two different energy levels. The six highest peaks of spectrum with the height greater than $0.01$ are located at $\Delta E= 0.1885, 0.5655, 0.8168, 0.9425, 1.3195, 2.0735$. 
In order to determine which two levels such transition happens between, we first use the {nine lowest energy levels $E_n$ with $n=0,1,\ldots,8$ obtained by the Pruess method with a relative error tolerance of $10^{-10}$\cite{PruessFulton1993,ShaoCaiTang2006}: $-0.1008$, $0.0837$, $0.6302$, $1.2200$, $1.9177$, $2.6915$, $3.5316$, $4.4300$, $5.3808$, to calculate all $36$ possible energy differences by $\Delta E_{mn} = |E_m-E_n|$ and fix $m_0,n_0$ such that $\Delta E_{m_0n_0}$ minimizes the distance $|\Delta E-\Delta E_{mn}|$ over all these $36$ candidates. 
If $|\Delta E-\Delta E_{m_0n_0}|$ is far less than all the others, say, the former is less than one-tenth of the latter, 
then we regard $\Delta E$ to be the energy difference for the transition between the $n_0$-th and the $m_0$-th levels, 
otherwise we regard the case to be indeterminable. All the levels to be transited for those six peaks are listed in left part of Table~\ref{tab:v4_energy} with the errors under $5\%$. 
For example, 
the primary peak corresponding to $\Delta E= 0.1885$ represents the energy transition between the first and zeroth levels with $\Delta E_{10}=0.1845$ and the error is no more than $0.4\%$, see the left part of the third row in Table~\ref{tab:v4_energy}. However, there exist two indeterminable cases: $\Delta E = 0.5655, 0.8168$,
see the numbers in italics of Table~\ref{tab:v4_energy}, where we have presented two nearest pairs of energy levels for each case. In order to further distinguish these two indeterminable cases, a longer simulation until
$t=1000$ is performed to raise the resolution of frequency from $1/100$ to $1/1000$ with the spectrum shown in the right plot of Fig.~\ref{fig-v4} as well as the resulting energy differences collected in the right part of Table~\ref{tab:v4_energy}. Ten highest peaks with the height greater than $0.005$ are located at $\Delta E = 0.1822, 0.5466, 0.5906, 0.6974, 0.7728, 0.8419, 0.8985, 0.9488, 1.3195, 2.0609$,
as shown in Table~\ref{tab:v4_energy}, all of which allow determinable energy levels to be transited with the errors under $0.3\%$.  Now we can observe there that the first indeterminable case is split into $\Delta E = 0.5466,0.5906$ and the second one split into $\Delta E = 0.7728, 0.8419$ due to the higher resolution of frequency, that is, a low resolution of frequency leads to the indeterminable case when two energy differences are close; the numerical values for the transitions are all improved and two new transitions, $\Delta E=0.6974,0.8985$, are captured at the same time. 
The above analysis demonstrates that, 
without the prior knowledge of energy levels of the quantum system in question, 
our conservative spectral method combined with the standard spectrum analysis, 
is capable of capturing accurately the energy level transitions with the resolution of $1/t$ through a long time simulation until the final time $t$. Moreover, the variations of mass $\varepsilon_{\text{mass}}(1000)$ and energy $\varepsilon_{\text{energy}}(1000)$ are $1.1737\times10^{-6}$ and $3.2910\times10^{-6}$, respectively, 
and as predicted by Propositions \ref{pro:mass_con} and \ref{pro:energy_con}, enlarging the computational domain to cut down the boundary effect will further reduce both $\varepsilon_{\text{mass}}(1000)$ and $\varepsilon_{\text{energy}}(1000)$ to the machine epsilon. For example $\varepsilon_{\text{mass}}(1000)$ (resp. $\varepsilon_{\text{energy}}(t)$) becomes no more than $1.8795\times10^{-14}$ (resp. $6.5650\times10^{-13}$) even on a very coarse grid mesh $Q = 1$, $M = 5$, $N = 8$ when resetting $-x_L = x_R =20$, $-k_{\min} = k_{\max} = 25\pi$. That is, 
the conservation laws, Propositions \ref{pro:mass_con} and \ref{pro:energy_con}, can be numerically verified  in the presence of unbounded polynomial potentials.

Next, we perform some quantum tunneling tests in those eight double wells shown in Table~\ref{tab:4apot}. The initial Gaussian wave packet occupies the well on the left by setting $A=1/\pi$, $\sigma_1 = 1$, $\sigma_2 = 1$, $x_0 = -2$, $k_0 = 0.5$, and we will measure the partial mass of the Gaussian wave packet contained in the well on the right as did in \cite{GrabertWeiss1985} via 
\begin{equation}
  \label{eq:rate_r}
  P_r(t) = \iint_{[0,x_R]\times\mathcal{K}} f(x,k,t)\D x\D k.
\end{equation}
Accordingly, the tunneling rate is just $P_r$ 
because the conservative spectral method conserves the total mass which equals to one here. 
For the symmetric potentials $V_1 \sim V_4$, in order of increasing barrier height, see Table~\ref{tab:4apot},  
it is readily observed from the left plot of Fig.~\ref{fig:v4_tunnelingrate} that the quantum tunneling is more likely to happen (as $P_r$ increases) for the barrier with a lower height. However, things become more complicated for the asymmetric potentials $V_5 \sim V_{8}$ as the tunneling rate not only depends on the height of the left-hand barrier $h$, but also depends on that of the right-hand barrier $h-g$. 
As shown in the right plot of Fig.~\ref{fig:v4_tunnelingrate}, the tunneling is more likely to happen for $V_6$ than $V_5$ when $t>6$, even though $V_6$ has a larger barrier height than $V_5$ (see Table~\ref{tab:4apot}, $9.9553 > 8.0465$). In fact, once the wave packet travels across the potential barrier with a certain probability, it is more likely to be trapped in the local minimum when $h-g$ is sufficiently large ($4.7230 < 6.0443$). However, the tunneling effect is still limited when the barrier height is too high, as it will lead to a relatively small probability of the wave packet to surmount the barrier (such as $V_7$). An extreme case is $V_8$, in which the barrier height $h - g$ on the right-hand well is too low ($0.0022$) and consequently it is easy for the wave packet to return to the well on the left. That accounts for the fact that the tunneling rate of $V_8$ seems the smallest although the barrier height $h$ of the left-hand side barrier is quite small (7.0455).
Moreover, one can see that the curves of $P_r$ in Fig.~\ref{fig:v4_tunnelingrate} are oscillatory after some time instants, and the oscillation under the asymmetric potentials is more violent than that under the symmetric ones
which may be caused by the higher barrier height of the former, see Table~\ref{tab:4apot}.
Such oscillation, as a typical quantum phenomena \cite{Zurek1991},  emerges directly from the oscillating structure of the Wigner function around the center barrier $x=0$, see e.g., Fig.~\ref{fig:p4_v7_time}. We can easily observe there that a highly oscillatory pattern appears when the wave packet tries to tunnel through the high barrier of $V_7$ ($h = 10.5165$) while most of the wave packet is left in well on the left. In contrast, once the wave packet travels across a lower barrier in a classical manner, then we should not expect such obvious oscillation and this is just the case for $V_4$ with the barrier height $h=0.40$. 
 \begin{figure}
     \includegraphics[width=0.5\textwidth,height=0.35\textwidth]{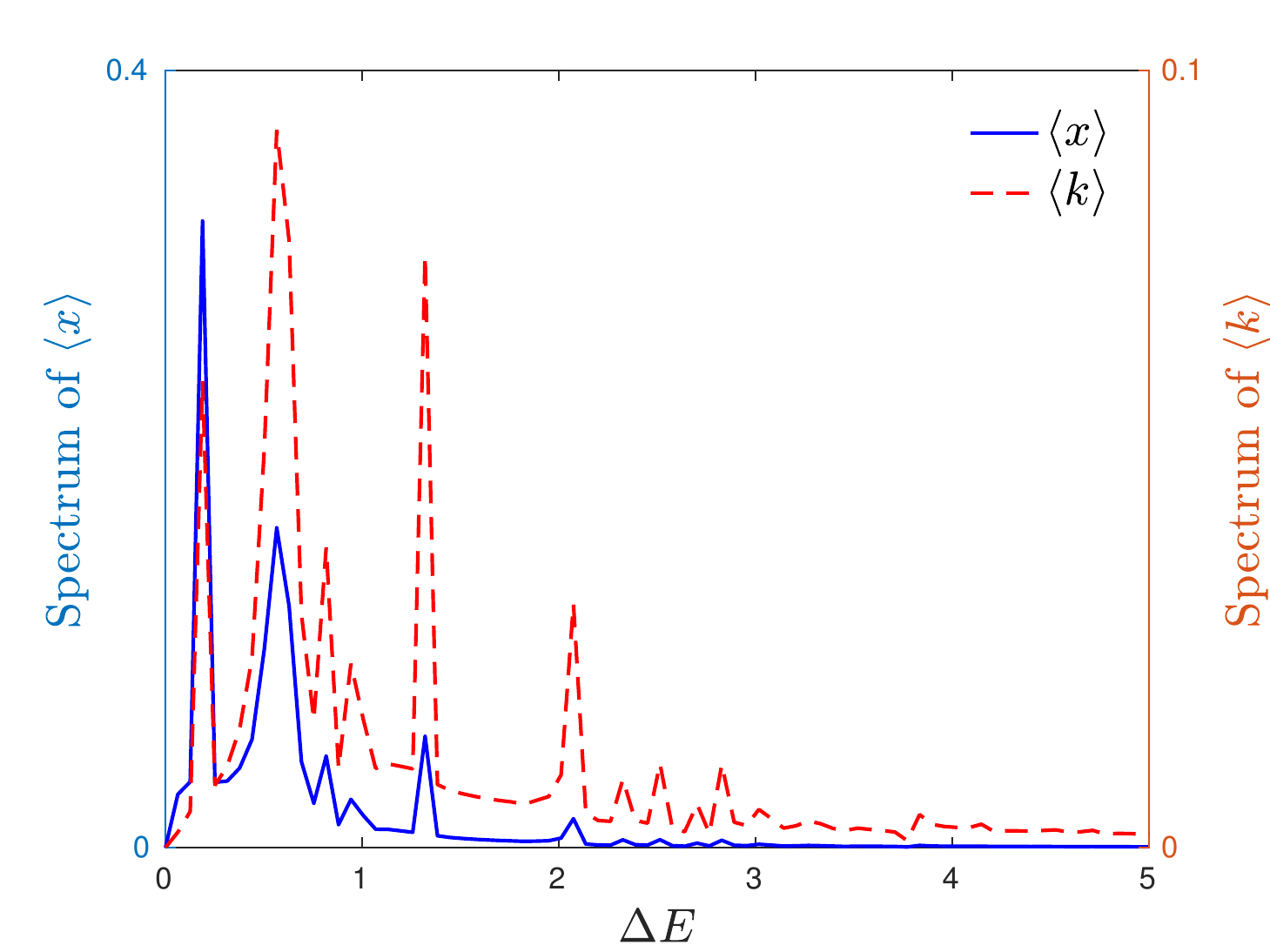}
     \includegraphics[width=0.5\textwidth,height=0.35\textwidth]{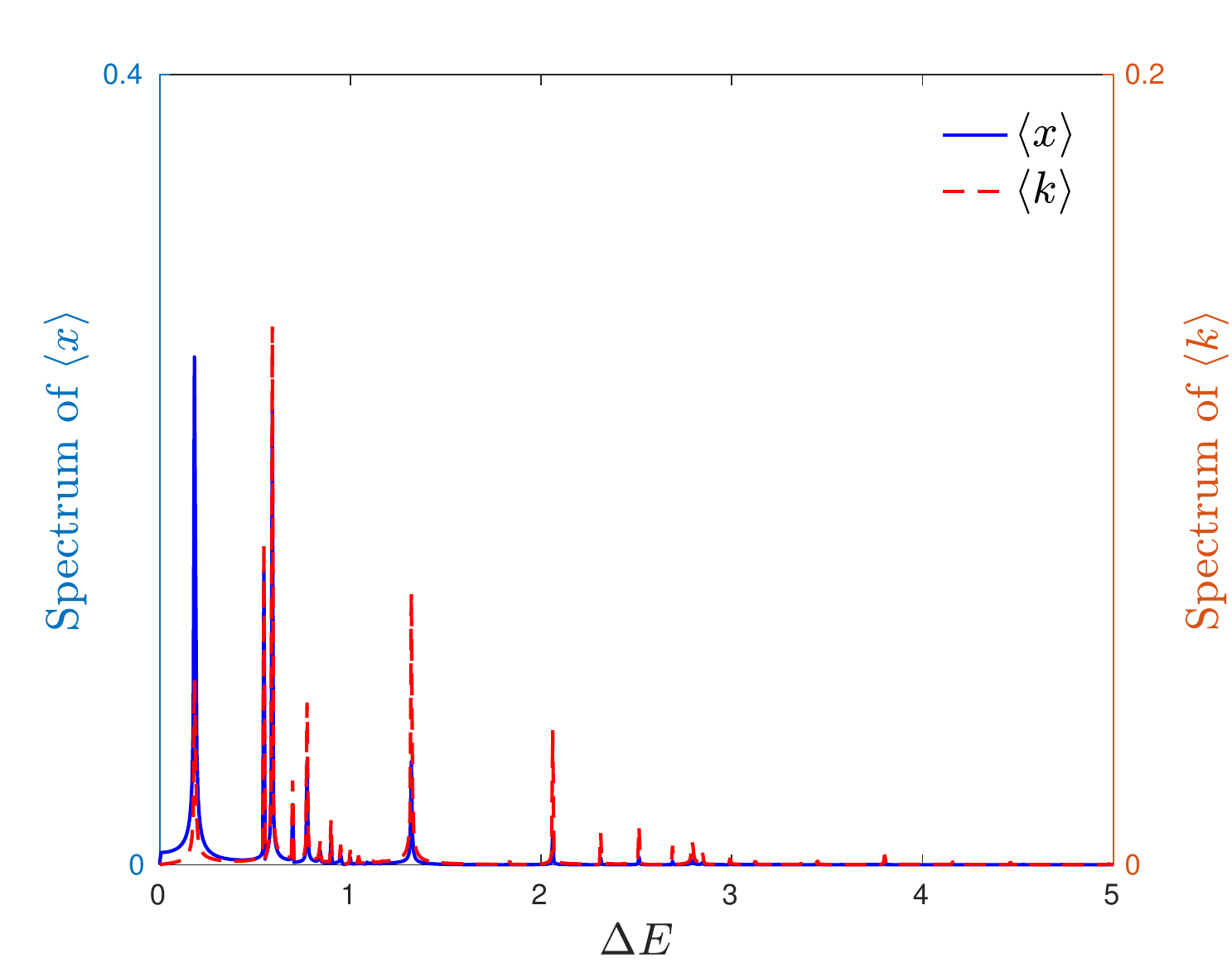}
   \caption{\small A fourth-order anharmonic oscillator with the symmetric potential $V_{4}$: Spectrum of the averaged displacement $\langle x\rangle$ and momentum $\langle k\rangle$ until $t=100$ (left) and $t=1000$ (right).}
 \label{fig-v4}
\end{figure}
\begin{figure}
   \includegraphics[width=.5\textwidth,height=0.35\textwidth]{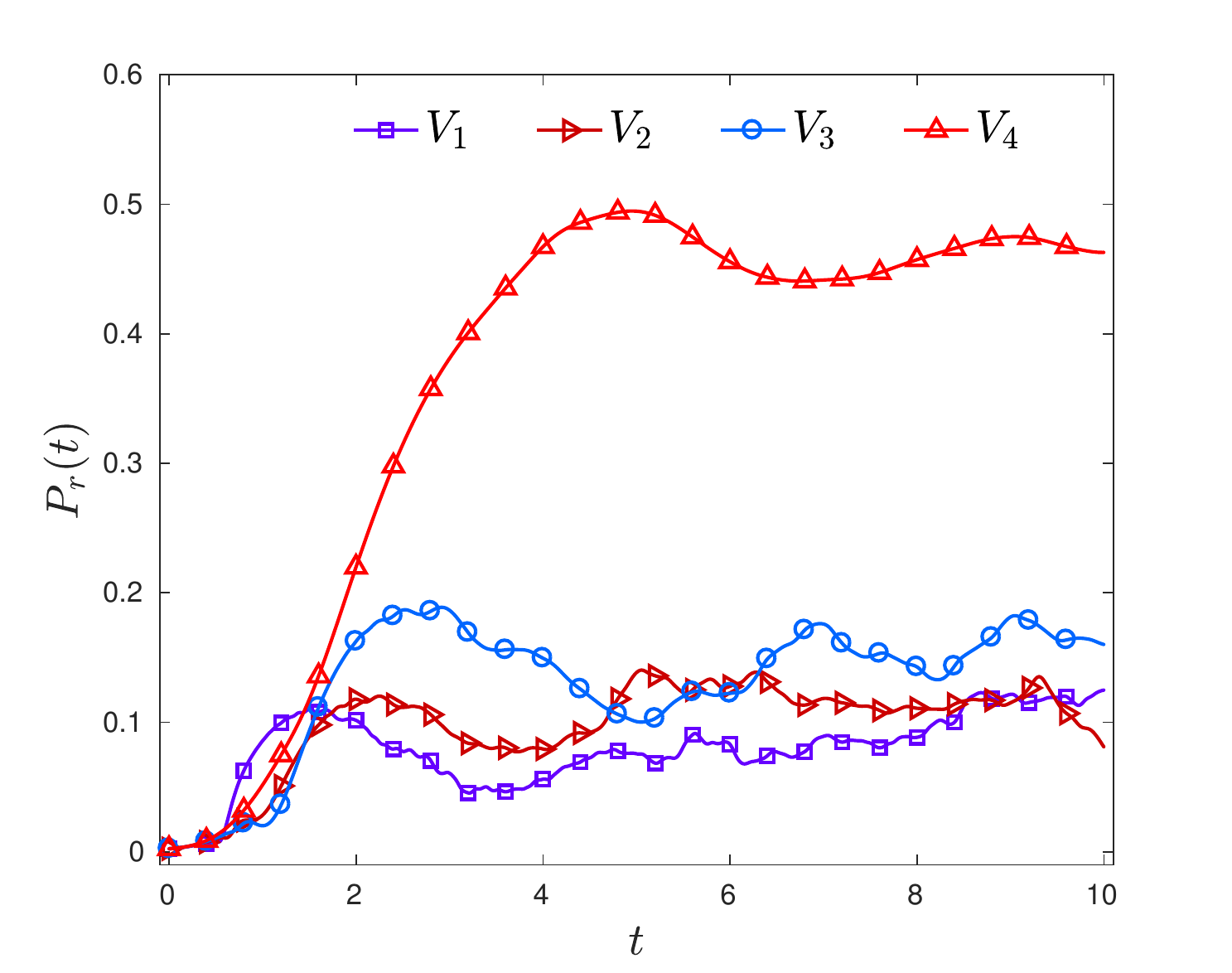}
   \includegraphics[width=.5\textwidth,height=0.35\textwidth]{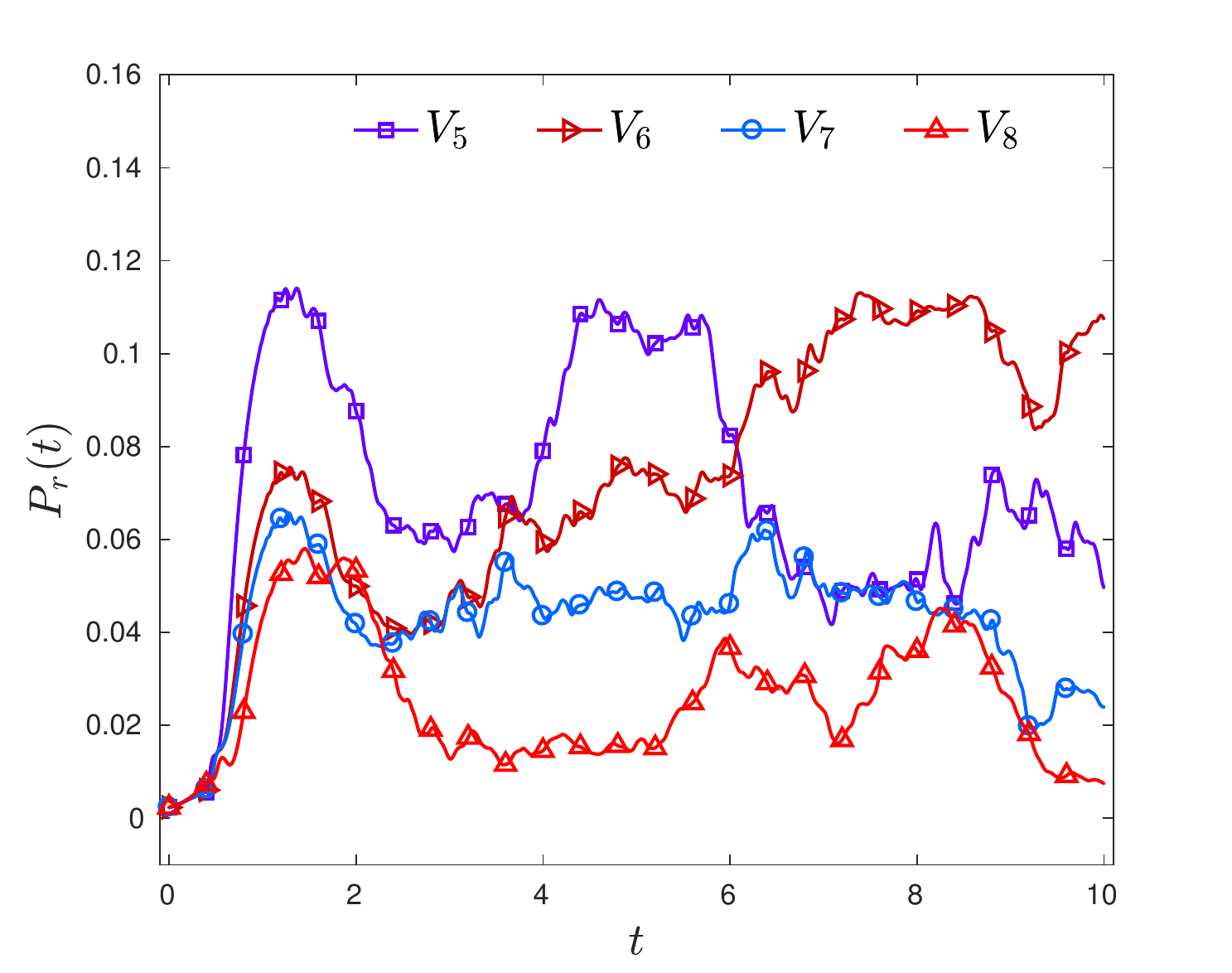}
 \caption{\small Fourth-order anharmonic oscillators: The quantum tunneling tests for symmetric (left) and 
 asymmetric (right) potentials. Here $P_r(t)$ represents the partial mass of the Gaussian wave packet contained in the well on the right at the instant $t$, see Eq.~\eqref{eq:rate_r}.}
\label{fig:v4_tunnelingrate}
\end{figure}
\begin{figure}[ht!]
  \subfigure[$t=0$.]{\includegraphics[width=0.32\textwidth,height=0.26\textwidth]{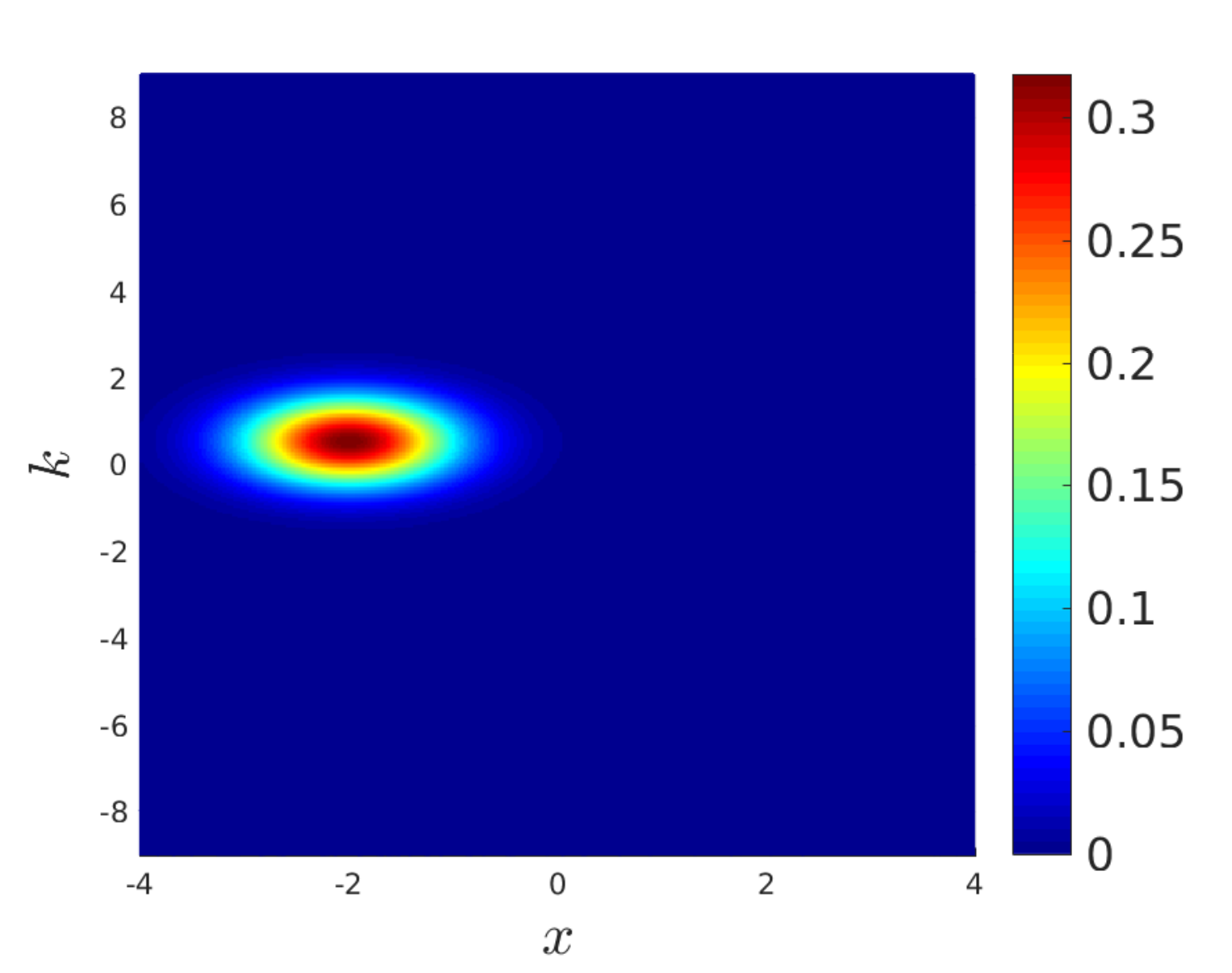}}
  \subfigure[$t=0.5$.]{\includegraphics[width=0.32\textwidth,height=0.26\textwidth]{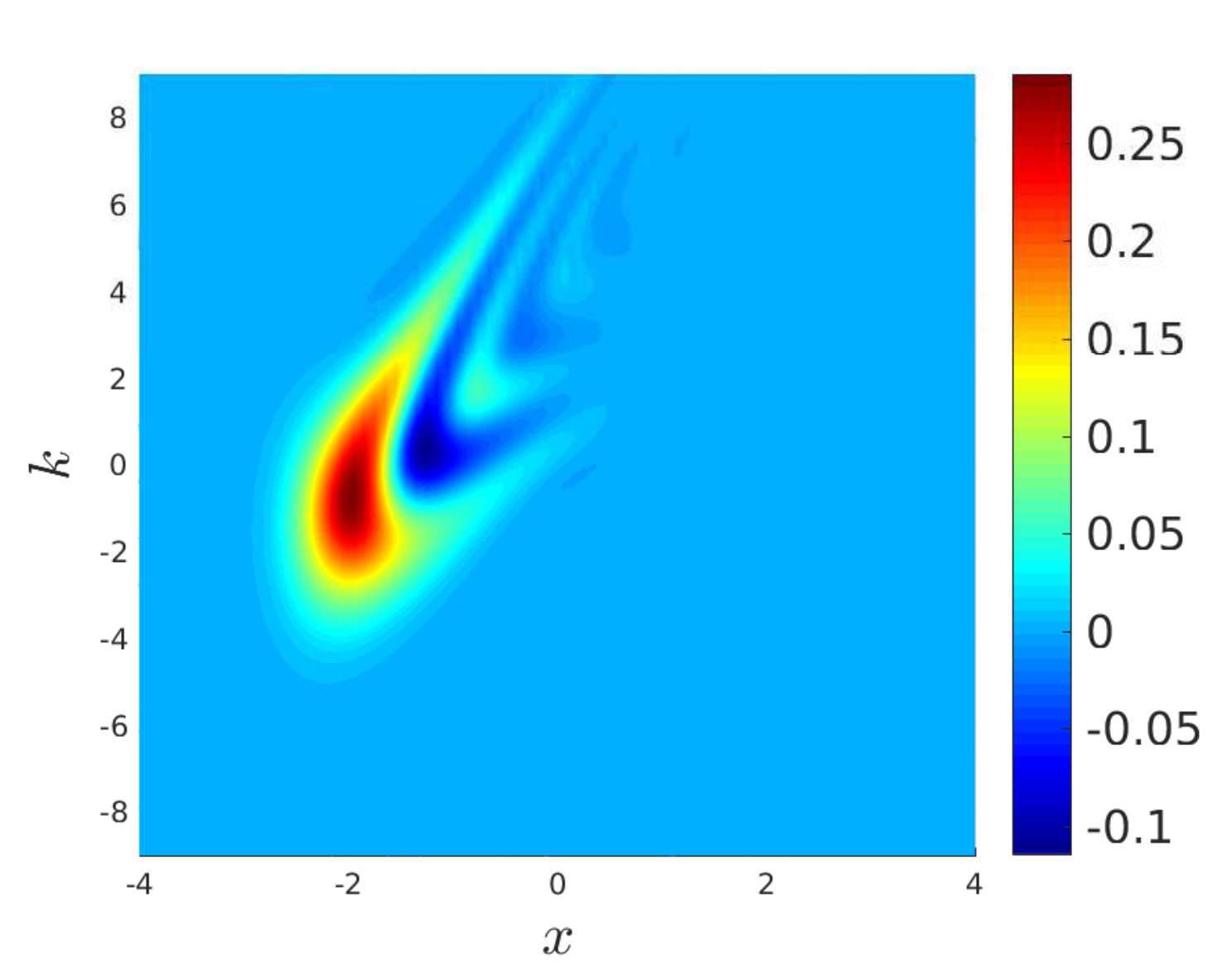}}
  \subfigure[$t=1$.]{\includegraphics[width=0.32\textwidth,height=0.26\textwidth]{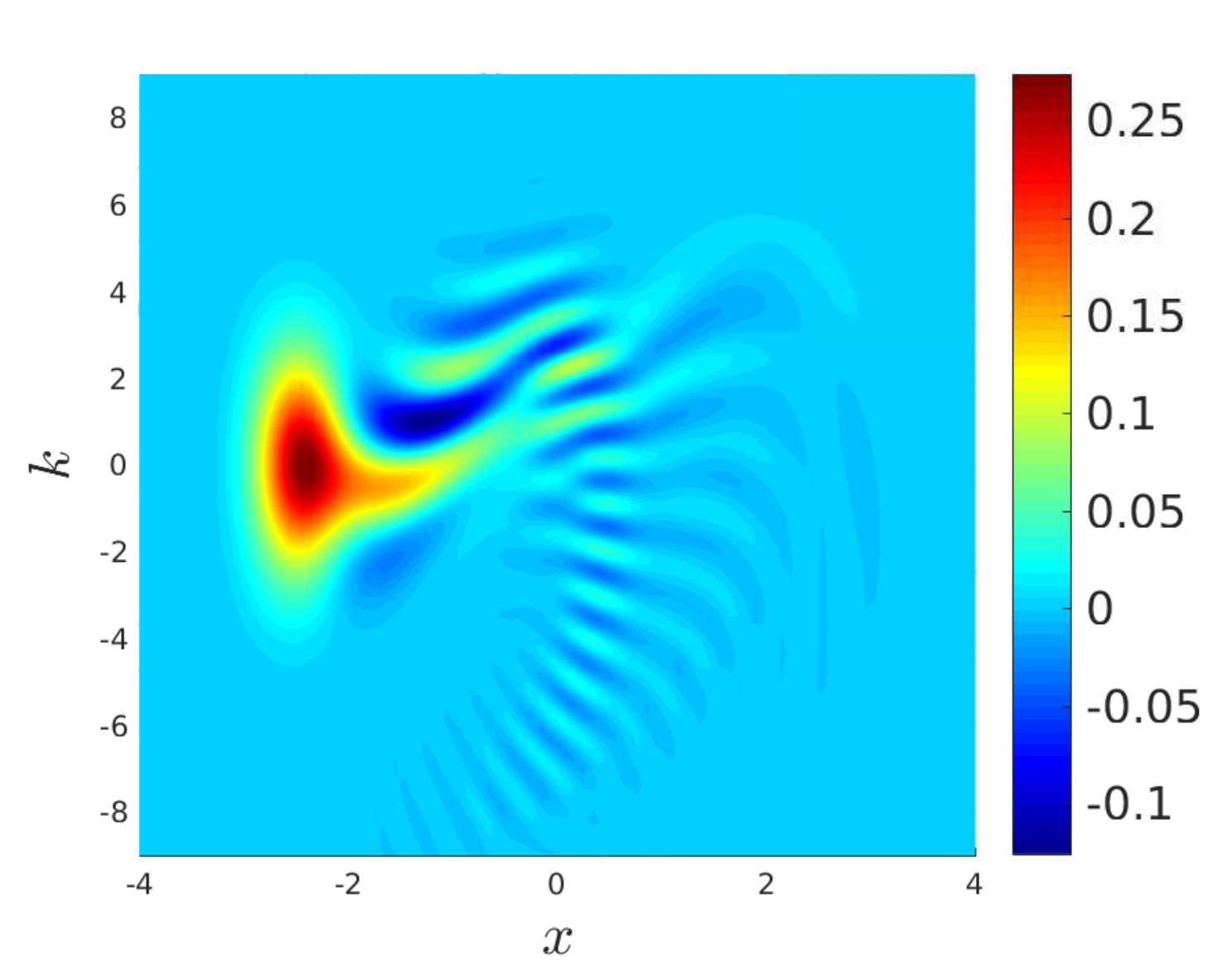}}

   \subfigure[$t=2$.]{\includegraphics[width=0.32\textwidth,height=0.26\textwidth]{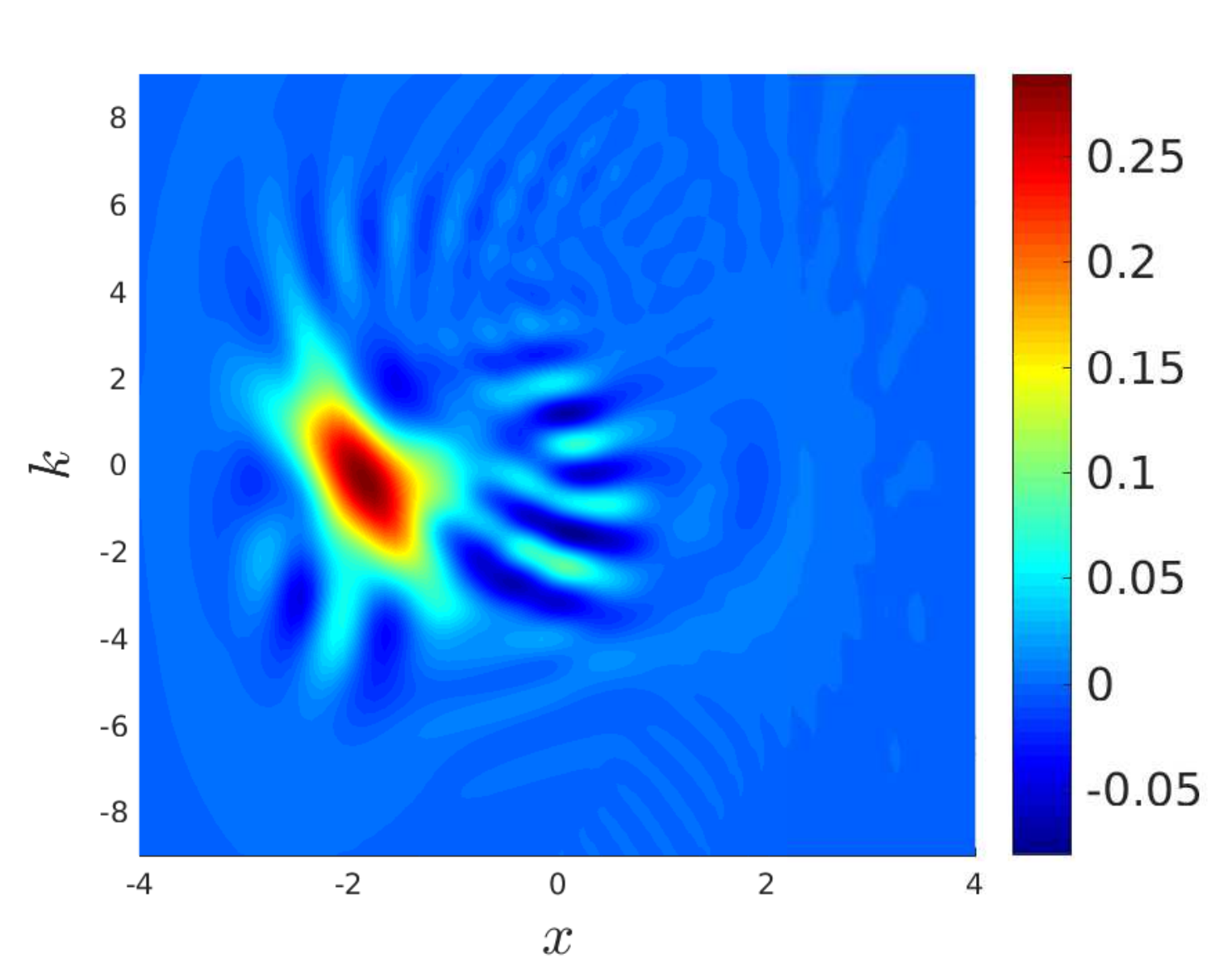}}
  \subfigure[$t=3$.]{\includegraphics[width=0.32\textwidth,height=0.26\textwidth]{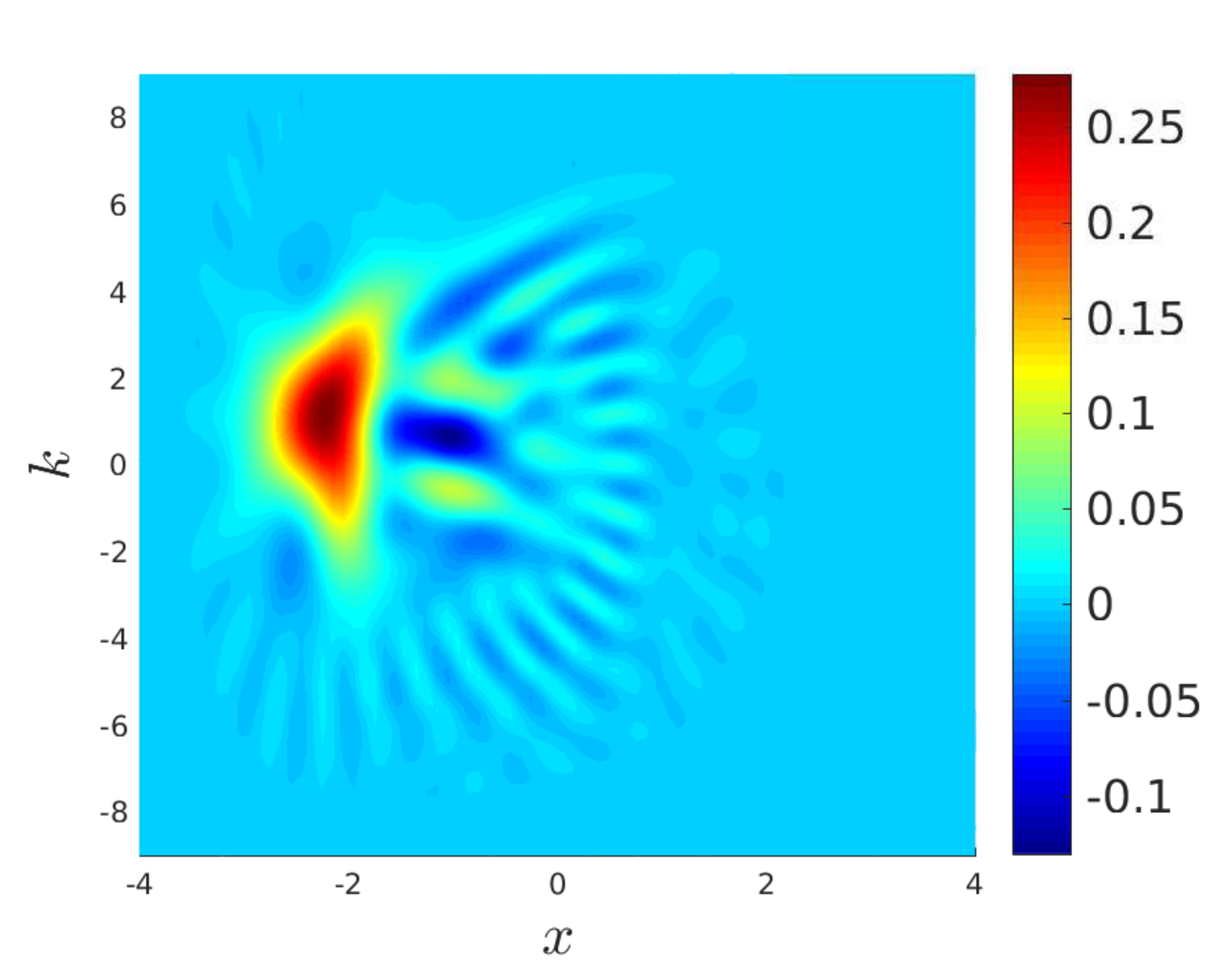}}
  \subfigure[$t=4$.]{\includegraphics[width=0.32\textwidth,height=0.26\textwidth]{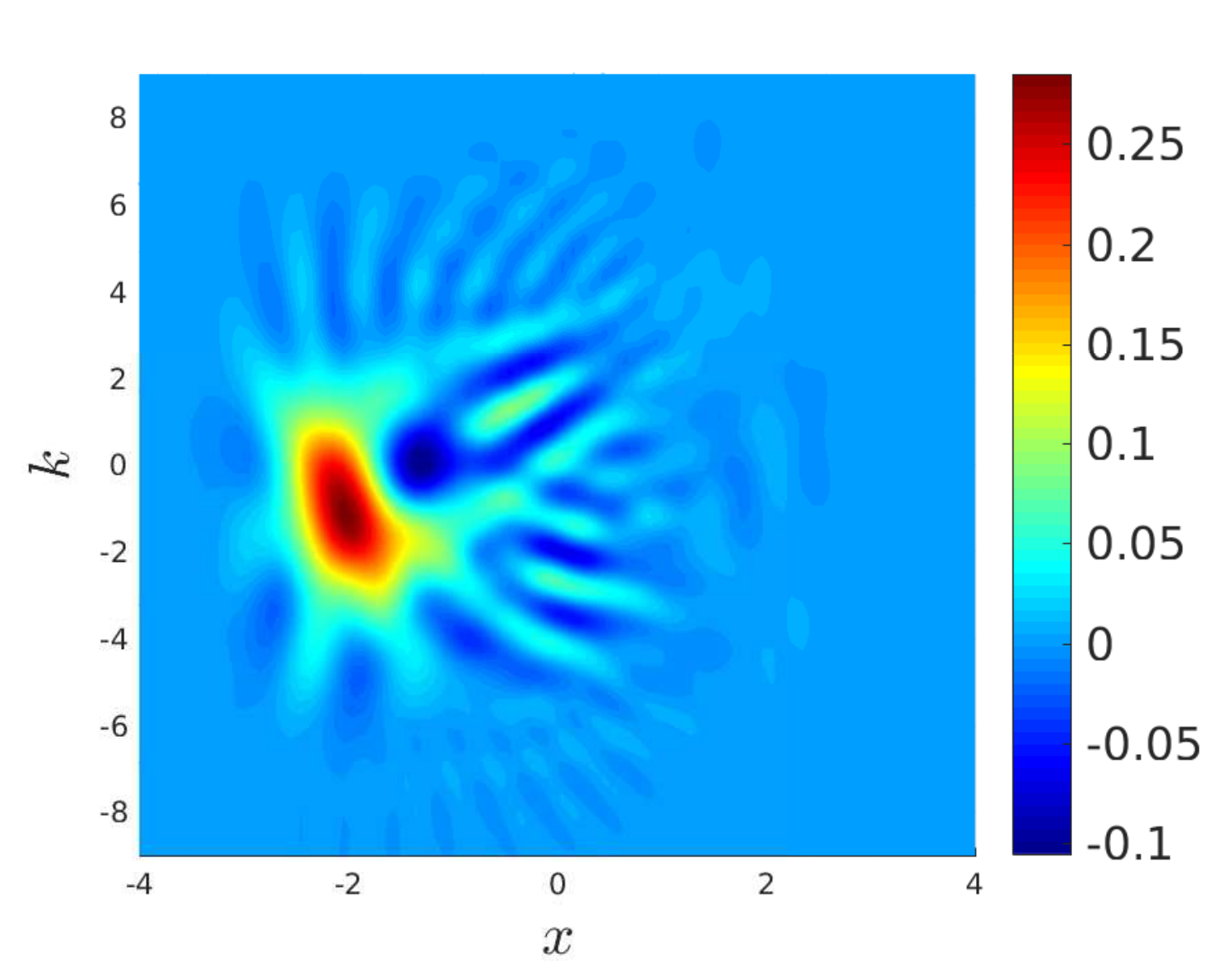}}

   \subfigure[$t=5$.]{\includegraphics[width=0.32\textwidth,height=0.26\textwidth]{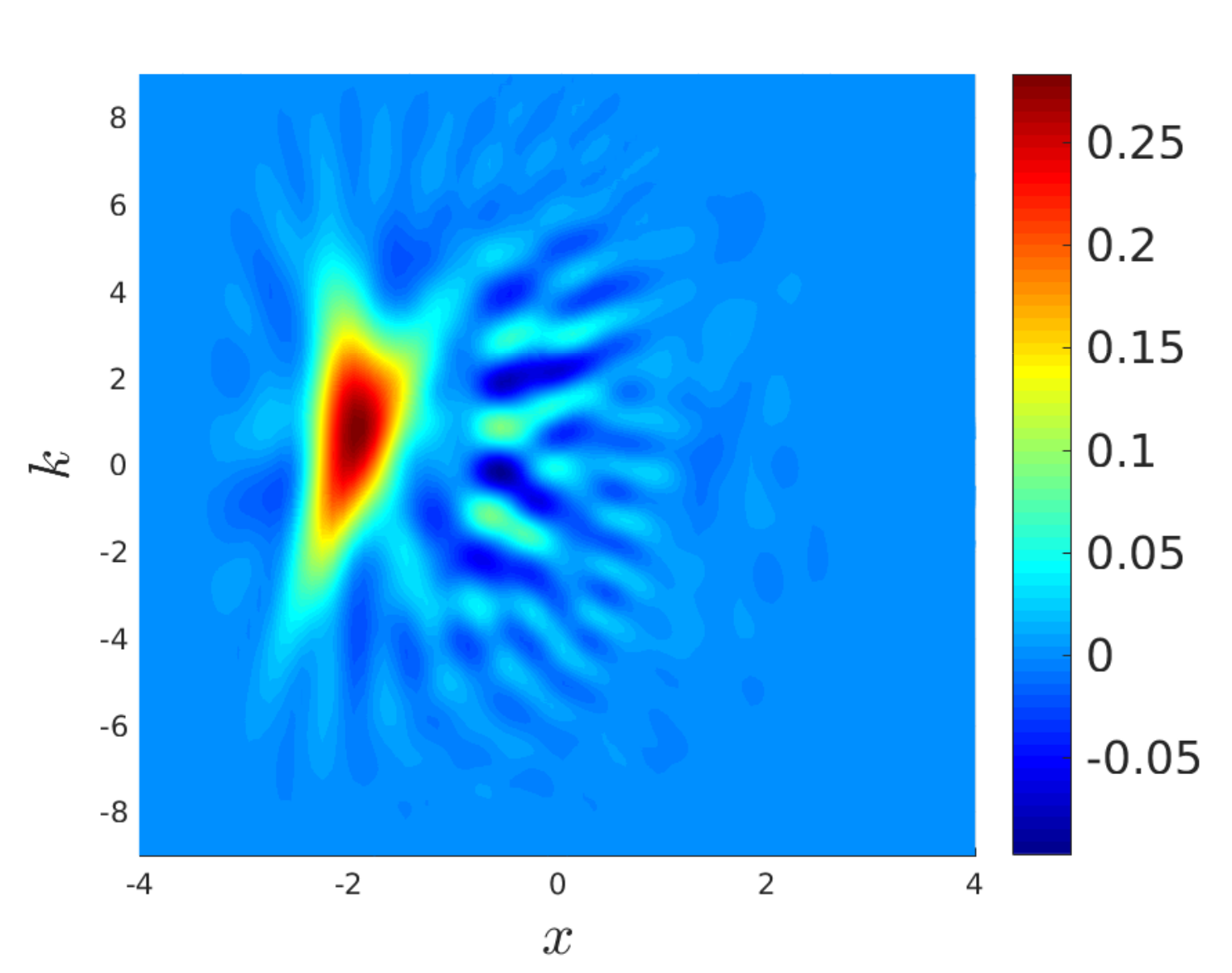}}
  \subfigure[$t=6$.]{\includegraphics[width=0.32\textwidth,height=0.26\textwidth]{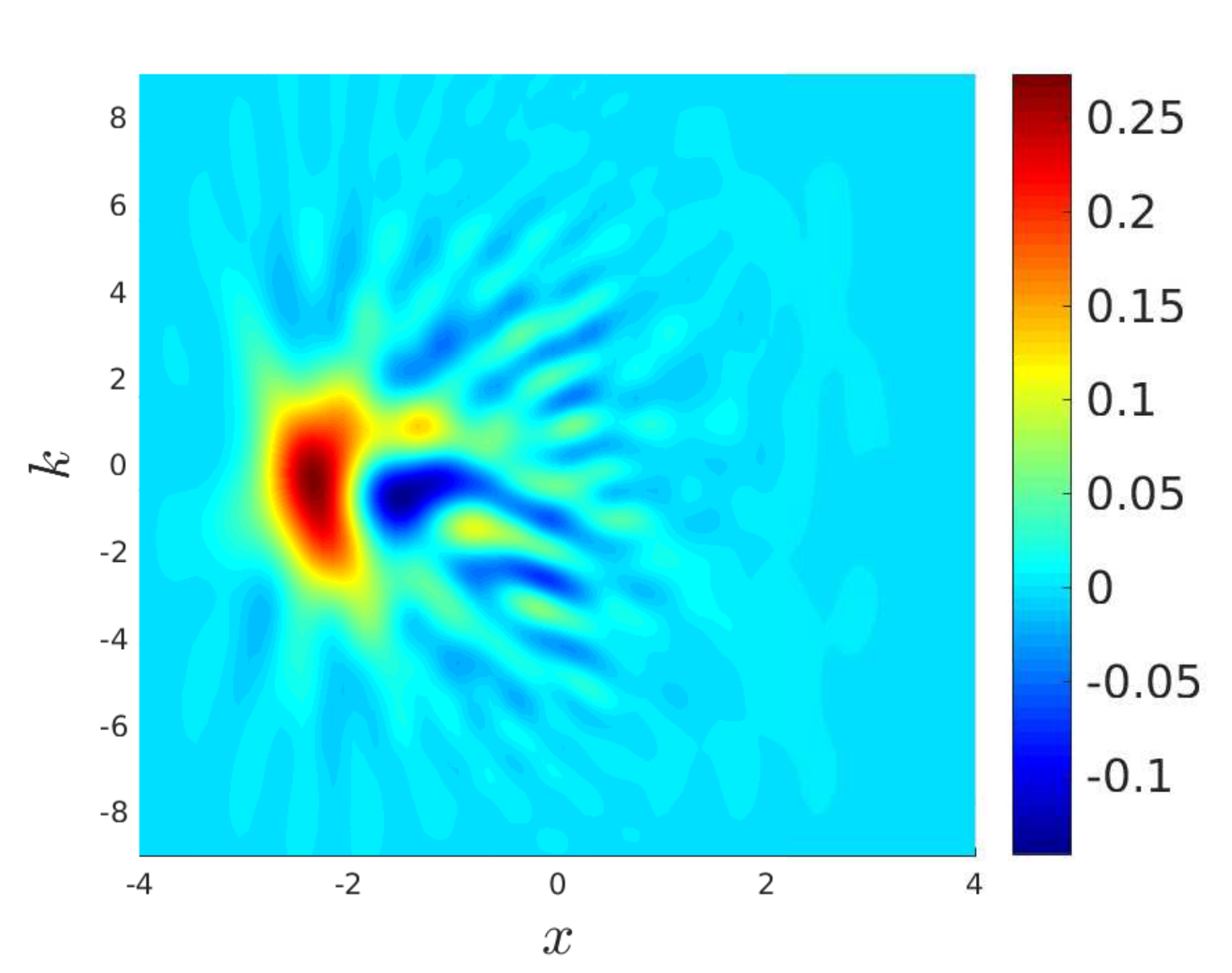}}
  \subfigure[$t=7$.]{\includegraphics[width=0.32\textwidth,height=0.26\textwidth]{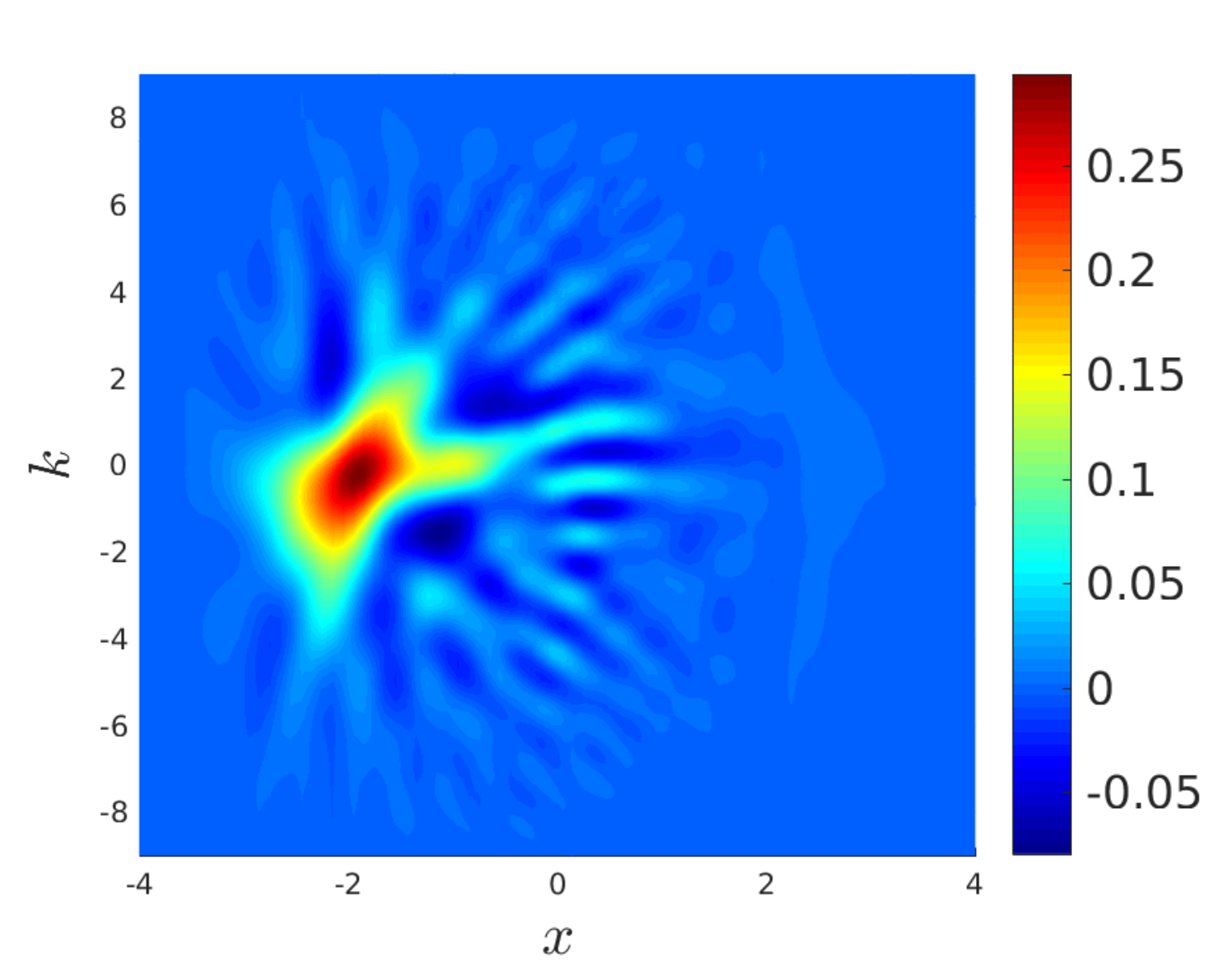}}

   \subfigure[$t=8$.]{\includegraphics[width=0.32\textwidth,height=0.26\textwidth]{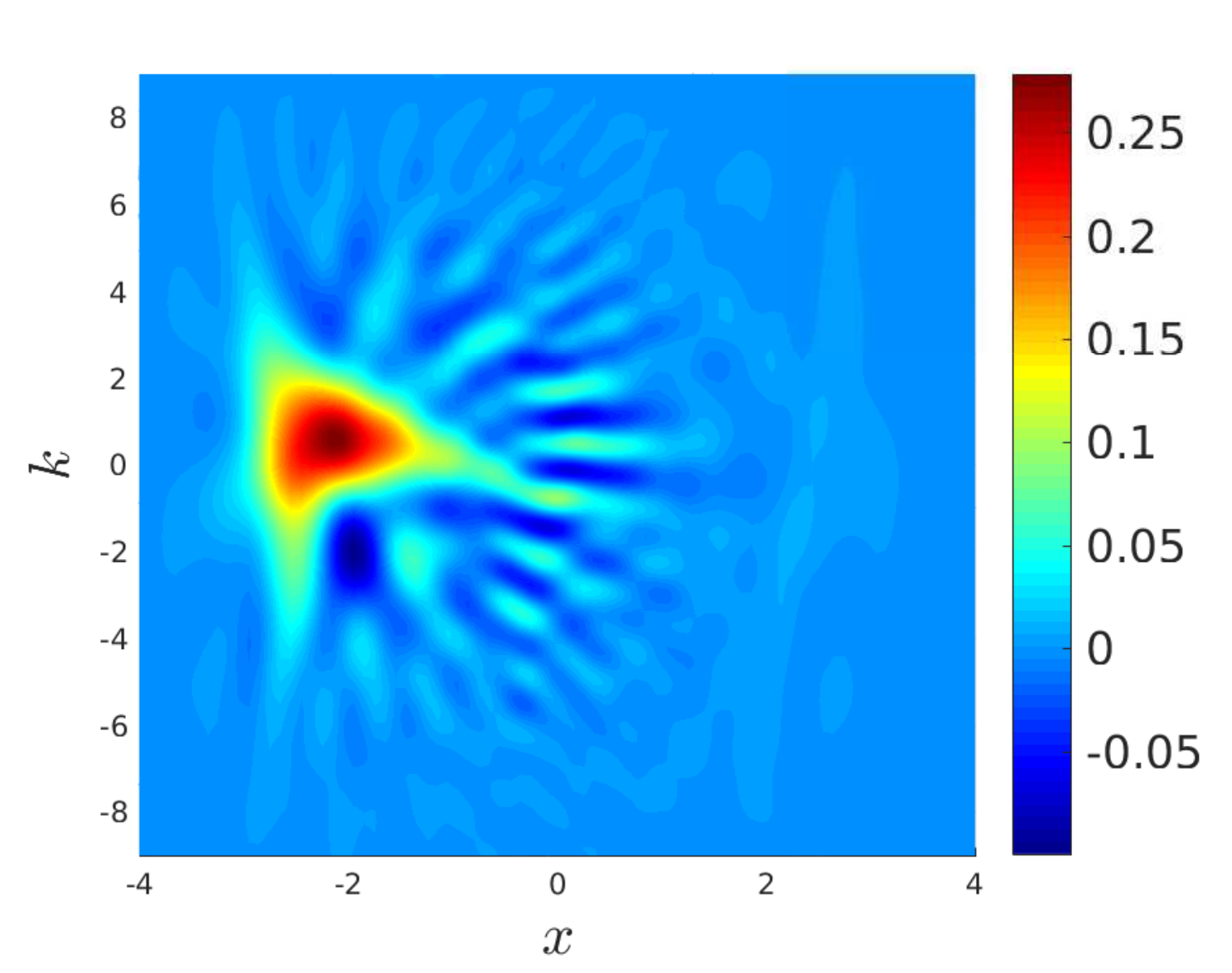}}
  \subfigure[$t=9$.]{\includegraphics[width=0.32\textwidth,height=0.26\textwidth]{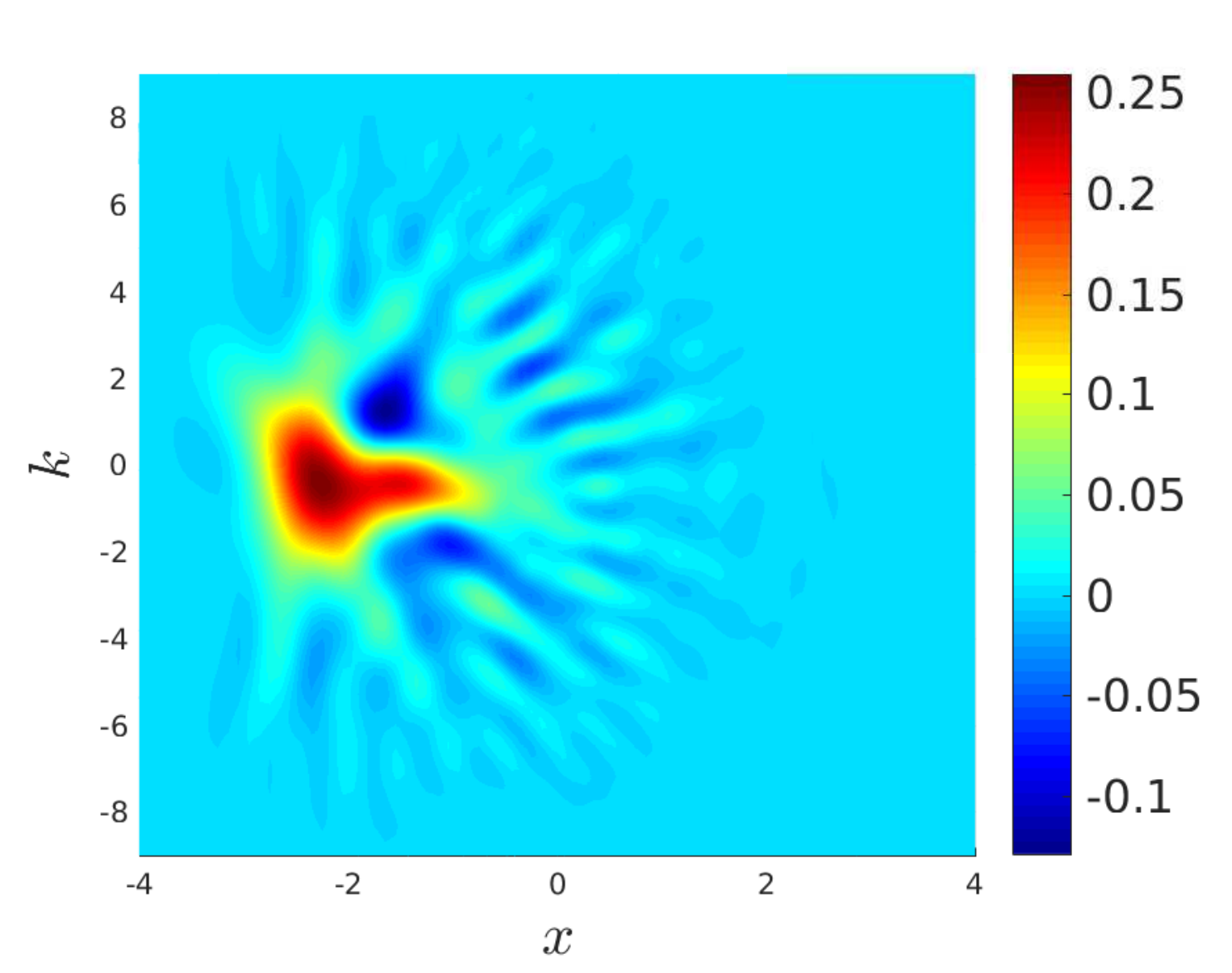}}
  \subfigure[$t=10$.]{\includegraphics[width=0.32\textwidth,height=0.26\textwidth]{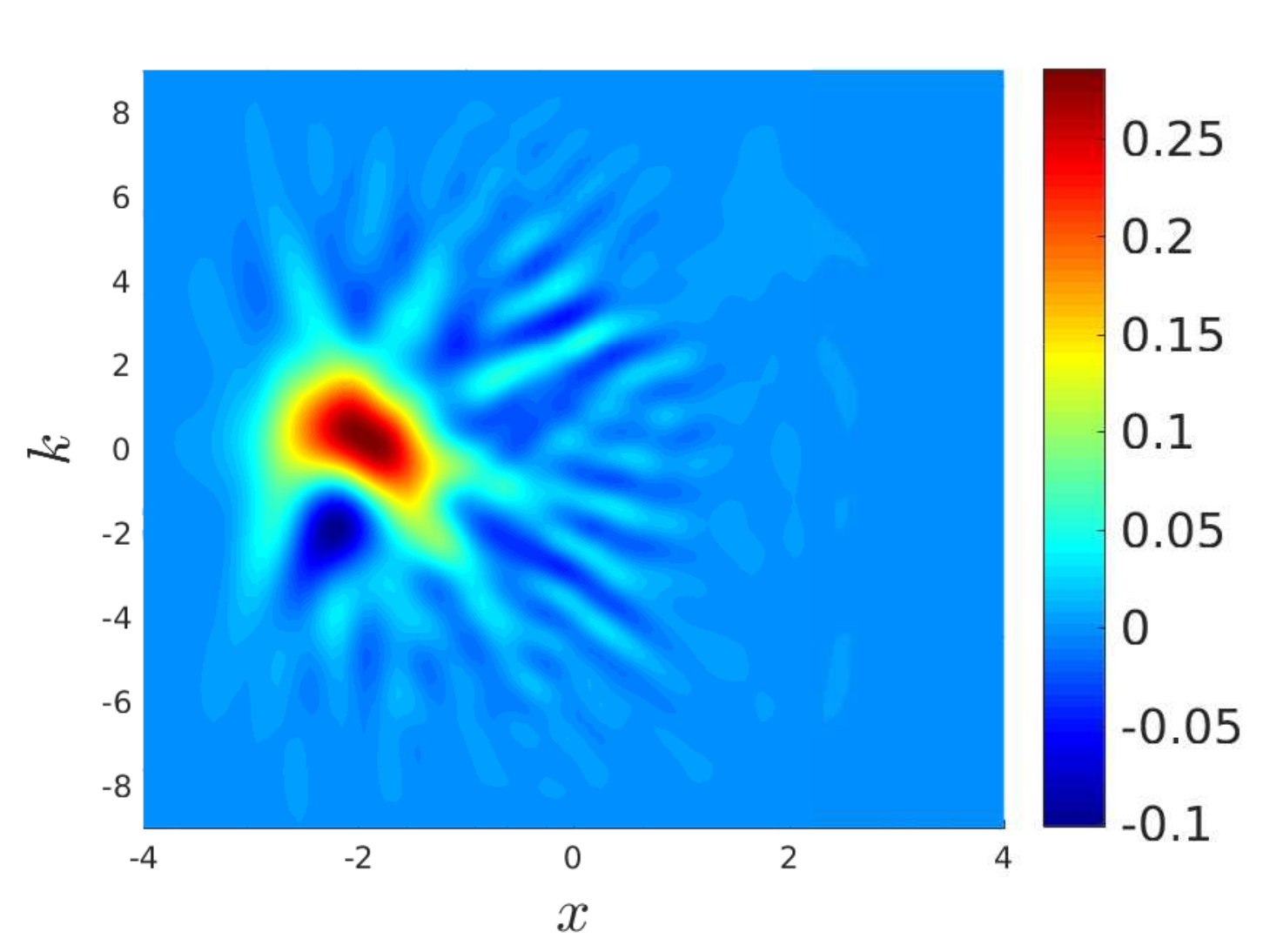}}
  \caption{\small Fourth-order anharmonic oscillators: the Wigner function at different instants $t=0,0.5,1,\ldots,10$ under the asymmetric potential $V_7$.}
  \label{fig:p4_v7_time}
\end{figure}

\subsection{Sixth-order anharmonic oscillators}
\label{sec:6_potential}

\begin{figure}
    \includegraphics[width=0.5\textwidth,height=0.35\textwidth]{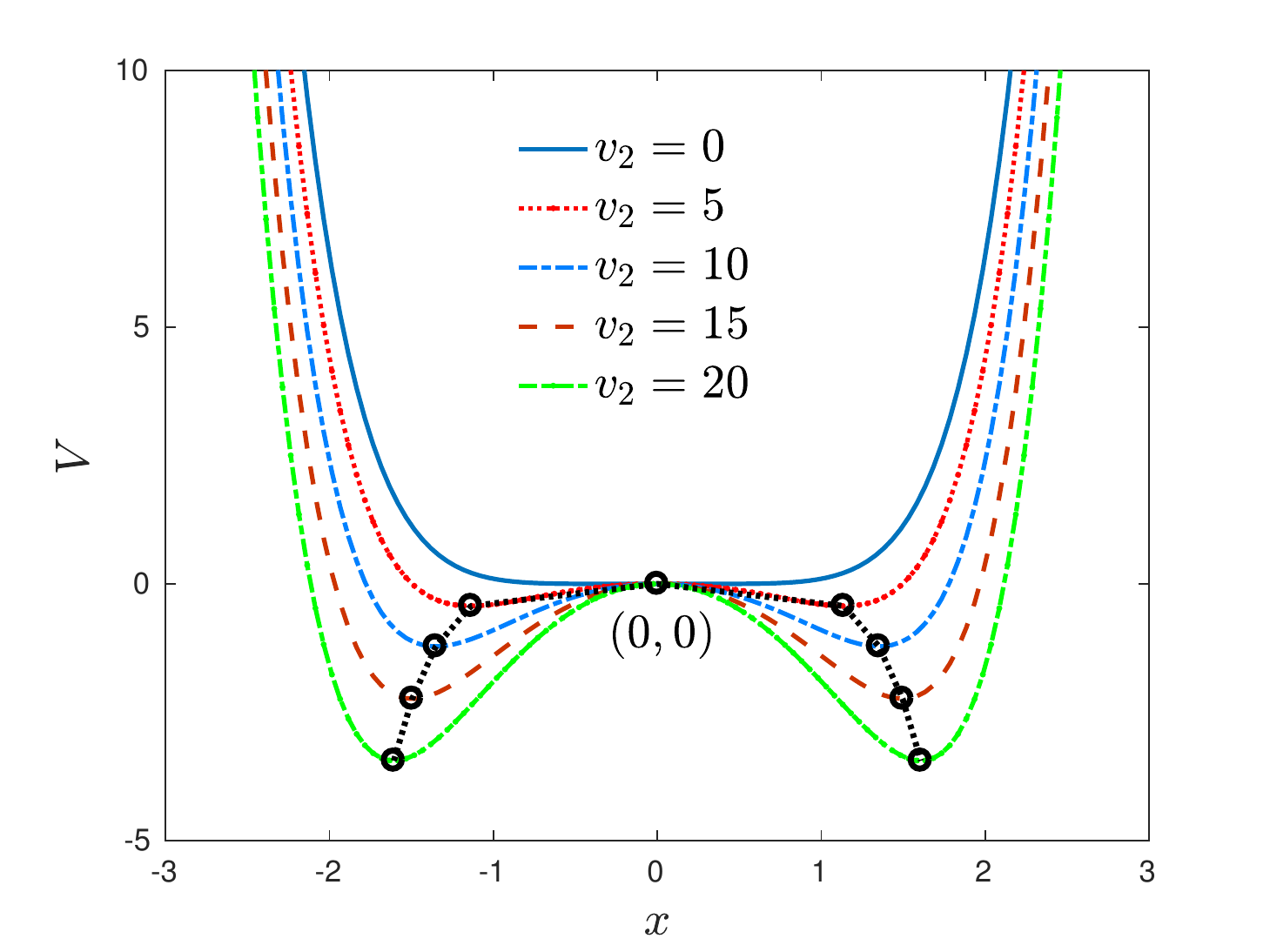}
    \includegraphics[width=0.5\textwidth,height=0.35\textwidth]{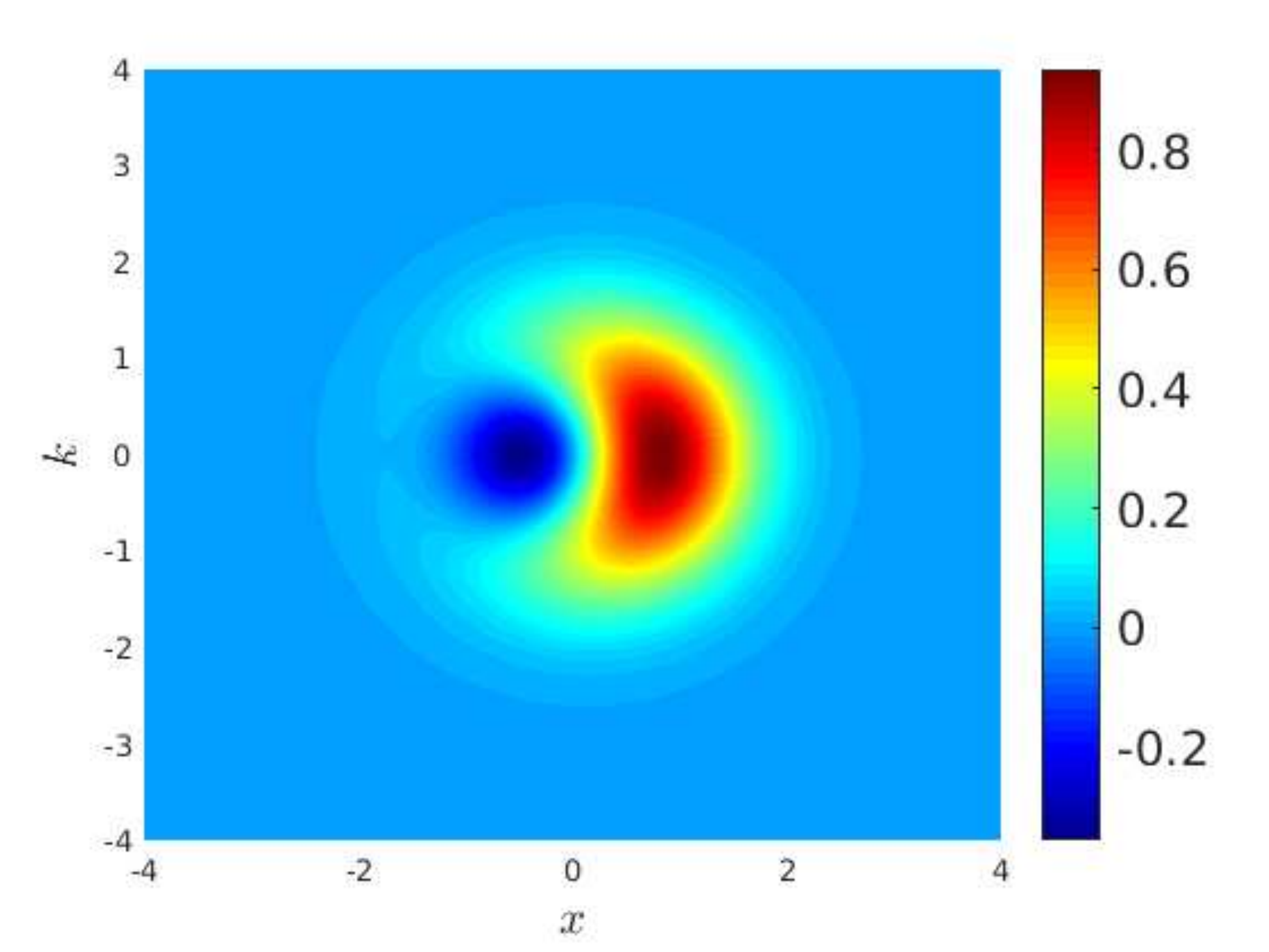}
 \caption{\small {The sixth-order anharmonic oscillators.  Left: The potentials $V(x)=\frac{1}{10}(x^6-v_2x^2)$ with $v_2=0, 5, 10, 15, 20$. The black circles mark the bottom of the wells. Right: The Wigner function $f(x,k) = \frac{1}{\pi}\left(\sqrt{2}x+x^2+{k^2}\right) \exp\left(-x^2-k^2\right)$.}}
 \label{fig:6pot}
\end{figure}

Sometimes a higher-order polynomial is needed to strongly confine the quantum system in a narrower region
like the sixth-order anharmonic oscillator \cite{Budaca2014}. We will show in this section that the proposed explicit conservative spectral solver is still capable of capturing the key quantum phenomena albeit a stricter time step 
must be used to deal with the stiffness introduced by the sixth-order polynomials. Let us consider the sixth-order anharmonic oscillators \cite{HeilbronnerRutishauserGerson1959,Gerson1961}:
\begin{equation}
  \label{eq:6_potential}
  V(x) = \frac{1}{10}\left( x^6 - v_2x^2\right),
\end{equation}
and the curves for $v_2=0, 5, 10, 15, 20$ are displayed in the plot of Fig.~\ref{fig:6pot}.
The initial Wigner function, located almost in the well on the right with $P_r(0)=0.8990$ as shown in the right plot of Fig.~\ref{fig:6pot}, is set to be a superposed state of the first and ground states of the harmonic oscillator with the potential $V(x)=x^2/2$,
and rotates periodically with a period of $2\pi$ \cite{SellierDimov2015-HO}.
The left plot of Fig.~\ref{fig:v6_uncertain} shows the partial mass $P_r(t)$ of the superposed state in the well on the right where other simulation parameters are set to be
$-x_L=x_R=10$, $-k_{\min}=k_{\max}=5\pi$, $\Delta t = 2\times 10^{-6}$, $N = 140$,  $Q = 10$ and $M = 21$.
It is readily observed there that the quasi-periodic rotations are evident in all cases and the periods are about 
$7.09$, $8.50$, $12.11$, $17.78$, $29.68$ for $v_2 = 0$, $5$, $10$, $15$, $20$, respectively, where we 
regard the Wigner function to complete a periodic rotation once the partial mass in the well on the right almost equals to $P_r(0)$. Compared with the period of $2\pi$ in the harmonic oscillator, i.e., the black line in the left plot of Fig.~\ref{fig:v6_uncertain}, 
the periods of the superposed state under the sixth-order double-wells are longer because of the existence of the central barrier. However, due to the quantum tunneling, the superposed state can still pass through the barrier and even the maximum mass in the right well is larger.
Next, we will investigate the Heisenberg uncertainty principle
\begin{equation}
  \label{eq:heisenberg}
   \sigma_x \sigma_p \geq \frac{\hbar}{2},
\end{equation}
where $\sigma_x$ and $\sigma_p$ are the standard deviations of position $x$ and momentum $p=\hbar k$, respectively.  
The initial value of $\sigma_x \sigma_p$ is $\sqrt{2}/2$, which also gives the minimum uncertainty of the harmonic oscillator.
During the Wigner quantum dynamics evolved by the explicit conservative spectral method,
we will measure  
\[
\sigma_x = \sqrt{\langle \left( x - \langle x\rangle \right)^2\rangle}, \quad \sigma_p = \sigma_{\hbar k} = \sqrt{\langle \left( k - \langle k\rangle \right)^2\rangle},
\]
where we have used $\hbar=1$, and the history curves of $\sigma_x \sigma_p - \hbar/2$ are displayed in the right plot of 
Fig.~\ref{fig:v6_uncertain}, where that for the harmonic oscillator is also plotted in back line for reference. 
It can be easily seen there that the uncertainty principle \label{eq:heisenberg} is definitely confirmed for all cases, 
the uncertainty increases generally as the barrier becomes higher, 
and the maximum values of $\sigma_x \sigma_p$ under the sixth-order double-wells are all much larger than that for the harmonic oscillator
while the minimum values all less than that for the harmonic oscillator. Fig.~\ref{fig:height} plots both maximum and minimum values of the uncertainty against the barrier height $h$, and shows that the maximum values of $\sigma_x \sigma_p$ are almost proportional to the barrier height whereas the minimum values keep almost the same. We may explain such phenomena as follows.
The quantum effect is
strengthened when a central barrier separates the wave packet and thus the maximum uncertainty increases with the height of the barrier.
On the contrary, a steeper potential (sixth-order polynomial) causes the superposed state of the harmonic oscillator (second-order polynomial) to be more local, i.e., the quantum effect is suppressed, and thus the minimum uncertainty becomes smaller.

\begin{figure}
    \includegraphics[width=0.5\textwidth,height=0.35\textwidth]{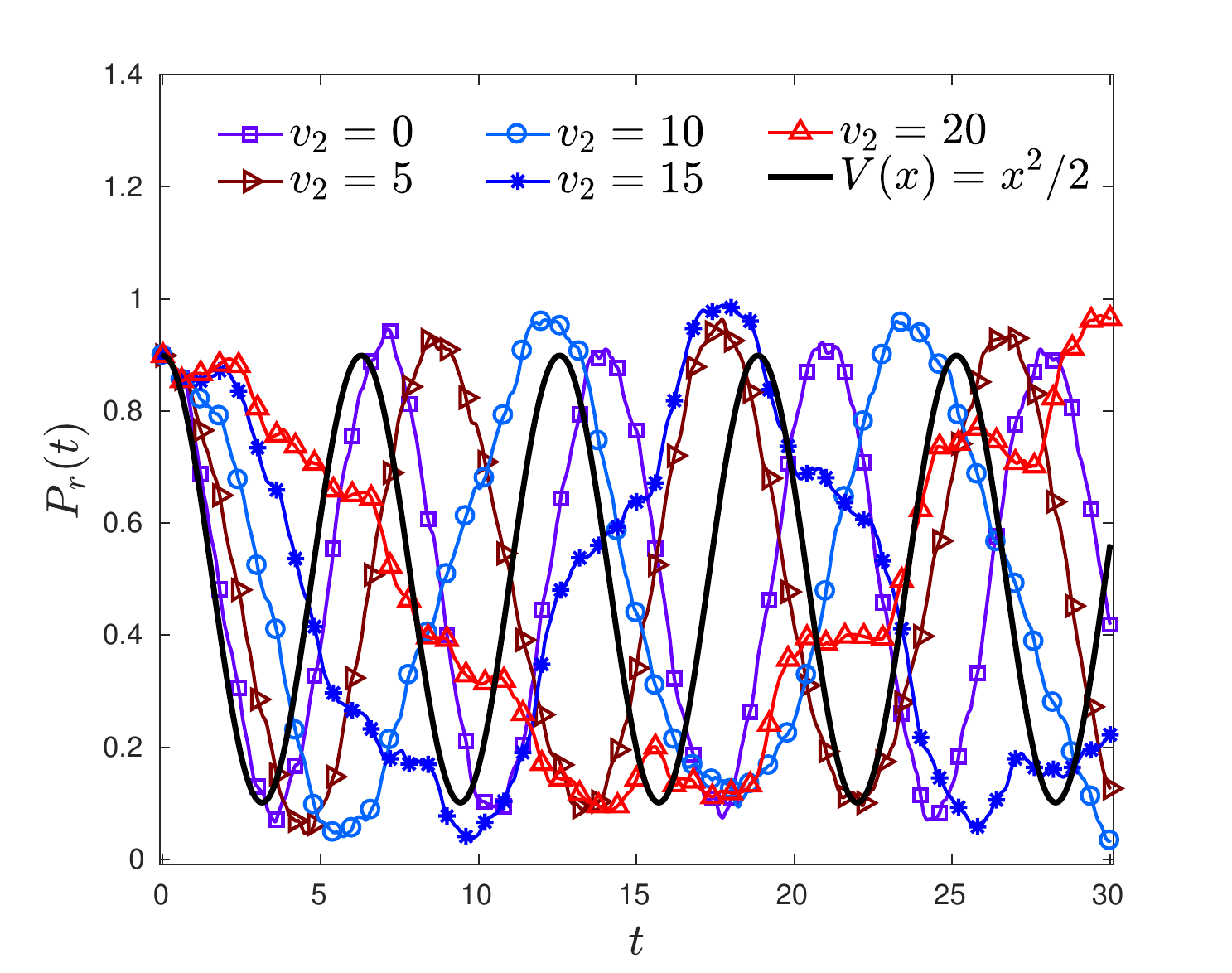}
    \includegraphics[width=0.5\textwidth,height=0.35\textwidth]{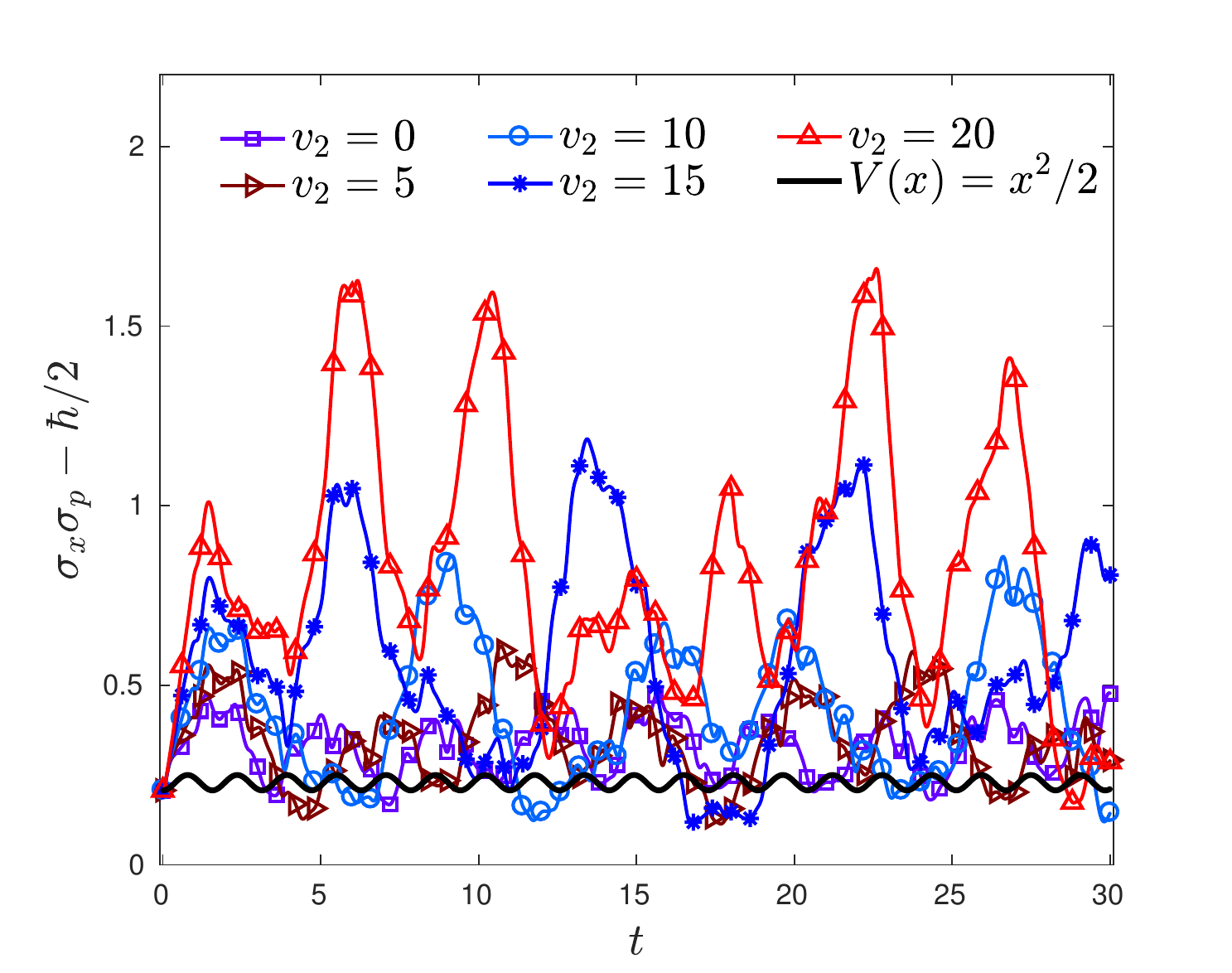}
 \caption{\small {The sixth-order anharmonic oscillators.  Left: The partial mass $P_r(t)$ of the superposed state in the right well. Right: The time evolution of the uncertainty and the uncertainly principle $\sigma_x \sigma_p \ge \hbar/2$ absolutely holds.}}
 \label{fig:v6_uncertain}
\end{figure}
\begin{figure}[ht!]
  \centering
    \includegraphics[width=0.6\textwidth,height=0.4\textwidth]{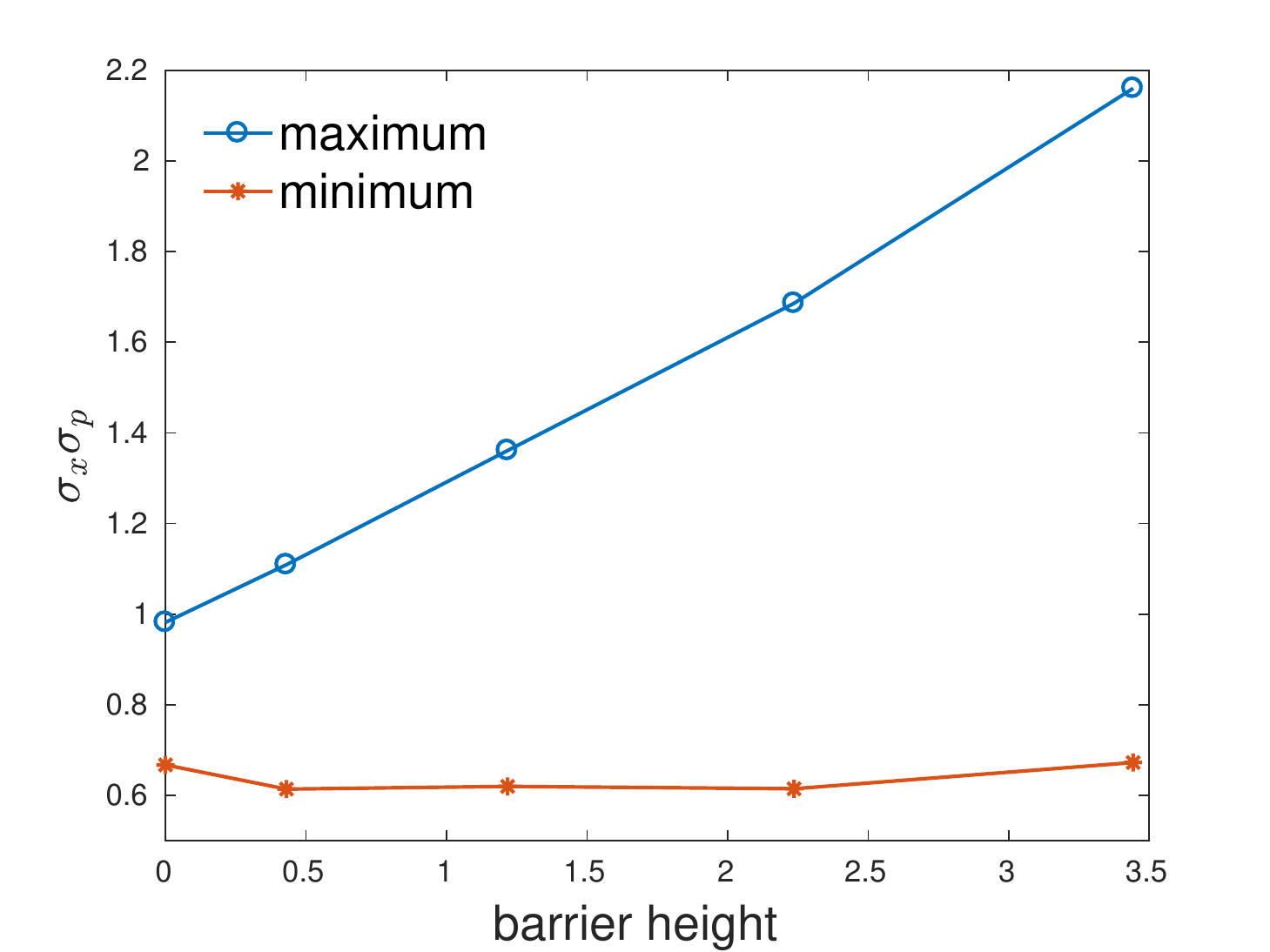}
 \caption{\small {The sixth-order anharmonic oscillators: The maximum and minimum values of $\sigma_x \sigma_p$ against the barrier height $h$ which increases as $v_2$ does (see the left plot of Fig.~\ref{fig:6pot}). The barrier height $h$ is $0$ for $v_2=0$, $0.4303$ for $v_2=5$, $1.2172$ for $v_2=10$, $2.2361$ for $v_2=15$, and $3.4427$ for $v_2=20$. Corresponding to these five barrier heights in increasing order,  the maximum values of $\sigma_x \sigma_p$ are $0.9818$, $1.1090$, $1.3609$, $1.6857$, $2.1601$, 
 and the minimum values $0.6673$, $0.6139$, $0.6199$, $0.6147$, $0.6725$. }}
 \label{fig:height}
\end{figure}

\subsection{An asymmetric double-well potential}
\label{sec:asym_dw}

Now we turn to consider a general class of double-well potentials, 
a mixture of both localized and polynomial potentials, see Fig.~\ref{fig:potential_split},
where the tunneling effects are largely determined by the localized fine structure and the unbounded polynomial is used to generate bound states. Apart from the partial mass $P_r(t)$ defined in Eq.~\eqref{eq:rate_r}, 
we also calculate the autocorrelation of the Wigner function \cite{KaczorKlimasSzydlowskiWoloszynSpasak2016,DavisHeller1984}:  
\begin{equation}
  \label{eq:correlation}
   C(t) = {\iint_{\mathcal{X}\times \mathcal{K}}} f(x,k, 0) f(x, k, t) \D x \D k,
\end{equation}
which characterizes the similarity between $f(x,k,0)$ and $f(x,k,t)$ as a function of the time lag between them and allows us to find repeating patterns, say, the latent periodic structures. That is, the correlation between the Wigner functions do not diminish or disappear over time, 
but oscillates on the frequencies proportional to $|E_m - E_n|$ after substituting Eq.~\eqref{eq:wigner_fun_pure} into Eq.~\eqref{eq:correlation}.

We first perform the accuracy check with the initial wave packet \eqref{init_GWP}: $A = 1/\pi$, $x_0 = 0$, $k_0=0.5$, $\sigma_1 = \sigma_2 = 1$ and other parameters: $-x_L = x_R = 15$, $-k_{\min} = k_{\max} = {10\pi}/{3}$, $Q = 20$, $\Delta t = 10^{-5}$. 
The spectral convergence with respect to both $N$ and $M$ can be clearly observed again for this unbounded potential in Fig.~\ref{fig:converg_adw} where 
the number of collocation points is fixed to be $N = 256$ for $k$-space (resp. $M = 31$ for each $x$-element) in studying the convergence rate with respect to $M$ (resp. $N$).

When investigating the quantum tunneling through the barrier with $h=1.6606$, the initial wave packet is relocated into the well on the left: $x_0 = -2$,
and two more faster moving ones are also considered: $k_0=1$ and $2$. The time history of the partial mass $P_r(t)$ on the mesh $(M,N)=(21,200)$ is shown in the left plot of Fig.~\ref{fig:adw_rate}. It can be easily seen there that 
the tunneling rate increases as expected when the initial kinetic energy $\hbar^2k_0^2/2m$ increases. In particular,
the slowest moving wave packet with $k_0=0.5$ can still penetrate the barrier partially though its initial kinetic energy, only $0.125$, is far less than the barrier height. This definitely manifests the power of quantum mechanics.  The results of autocorrelation function $C(t)$ are displayed in the right plot of  Fig.~\ref{fig:adw_rate}, and show that the magnitude of $C(t)$ decreases for the wave packet with higher initial kinetic energy whereas the peaks occur with almost the same frequencies since the oscillating frequencies must be proportional to $|E_m - E_n|$ and independent of the initial kinetic energy.

\begin{figure}[ht!]
   \includegraphics[width=0.5\textwidth,height=0.35\textwidth]{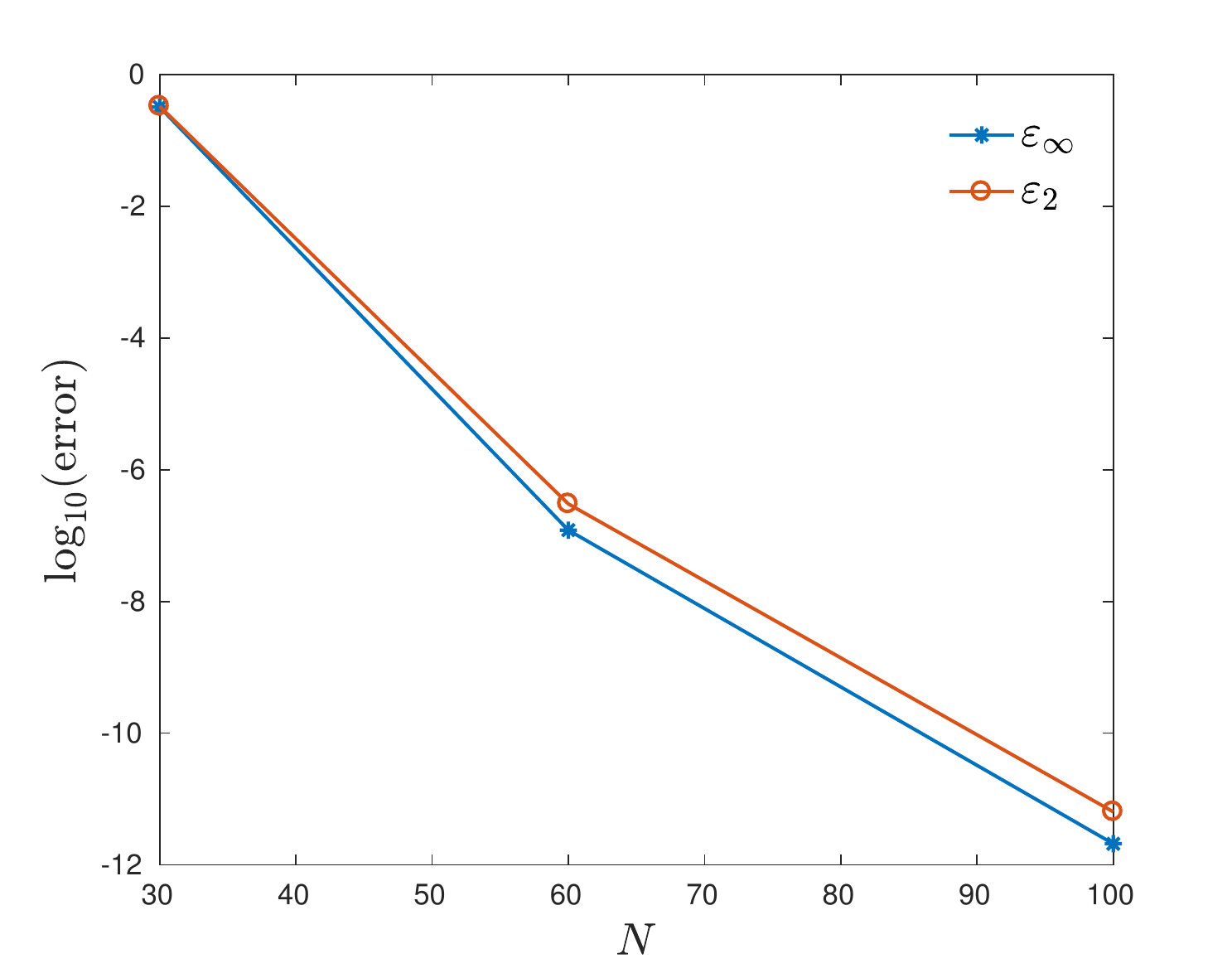}
   \includegraphics[width=0.5\textwidth,height=0.35\textwidth]{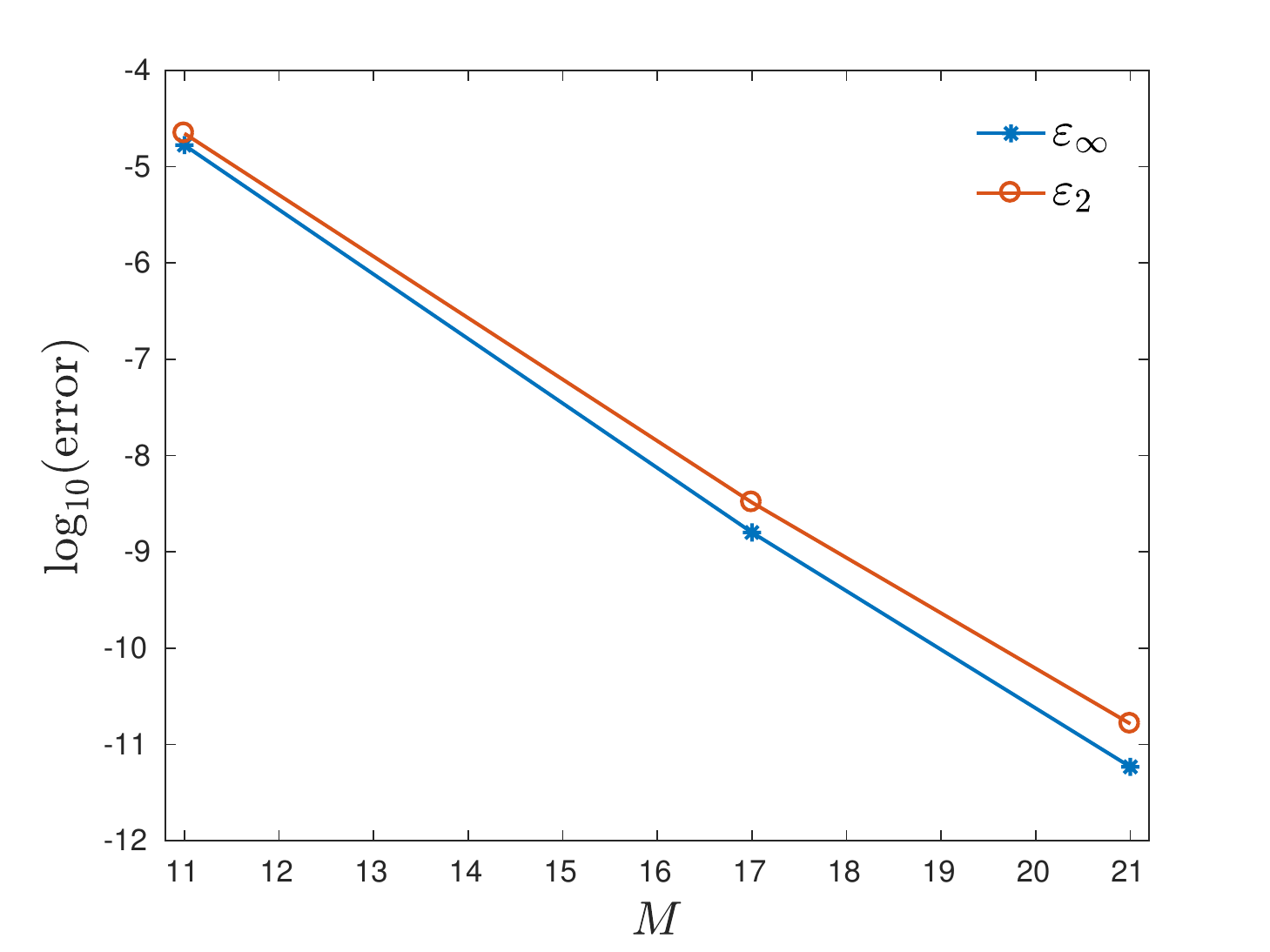}
   \caption{\small { An asymmetric double-well potential:  The convergence rate with respect to $N$ (left) and $M$ (right). The spectral convergence in both $x$-space and $k$-space is evident. All the errors are measured at the instant $t=10$.} }
   \label{fig:converg_adw}
\end{figure}

\begin{figure}[ht!]
   \includegraphics[width=0.5\textwidth,height=0.35\textwidth]{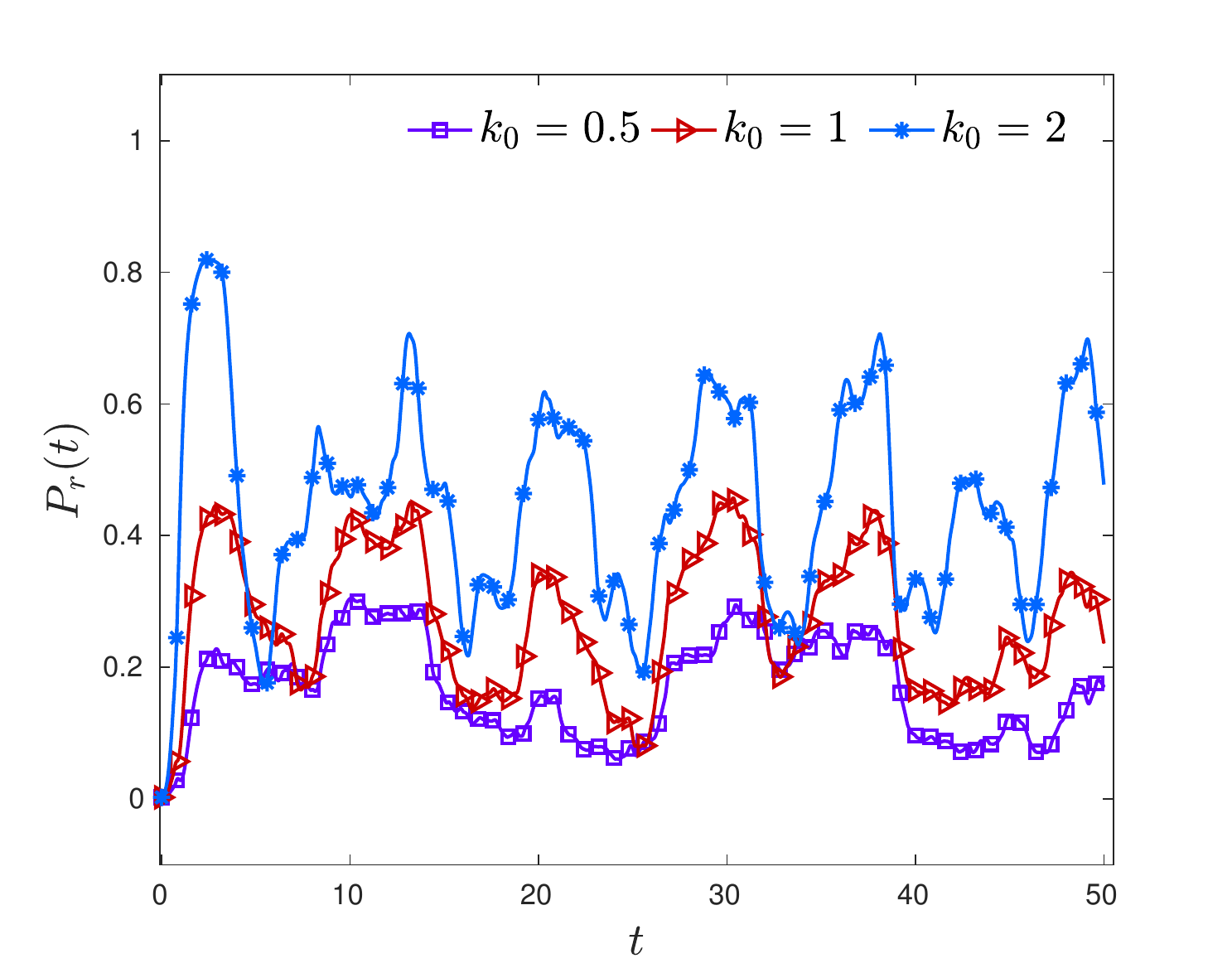}
   \includegraphics[width=0.5\textwidth,height=0.35\textwidth]{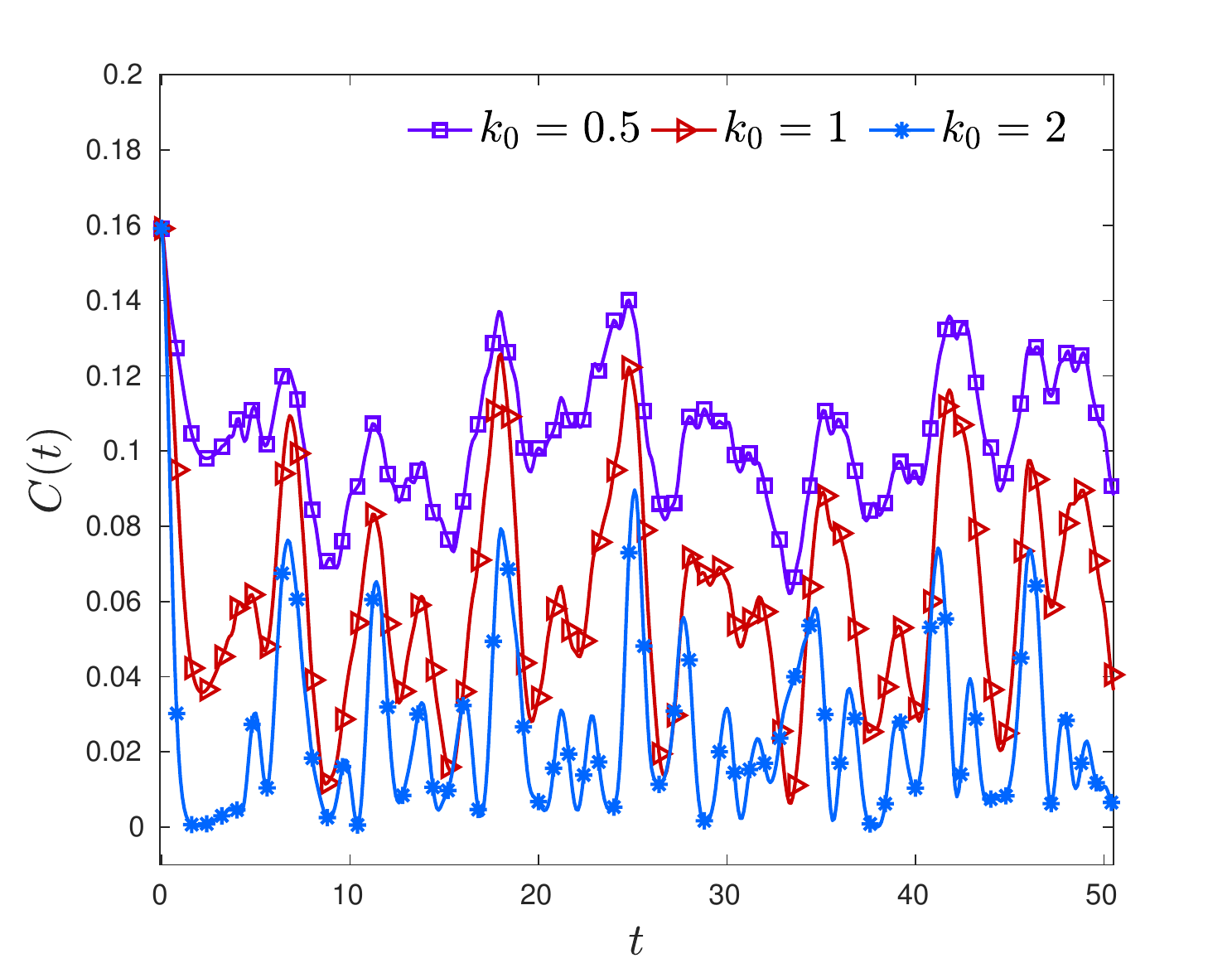}
    \caption{\small{ An asymmetric double-well potential: The partial mass $P_r(t)$ (left) and the autocorrelation function $C(t)$ (right). The tunneling rate is higher as the initial kinetic energy $\hbar^2k_0^2/2m$ increases. $C(t)$ oscillates on the frequencies proportional to $|E_m - E_n|$ which are independent of the initial kinetic energy.}}
    \label{fig:adw_rate}
\end{figure}

\subsection{A rational fraction function}
\label{sec:rational}

Finally, we will give a simple example to show that the operator splitting scheme may fail to conserve the energy. Consider the potential of the rational fraction \eqref{eq:fraction} and 
the resulting subproblem for $V_{pol}(x) = x^2-1$ can be solved analytically. The second-order Strang splitting is adopted here and the initial data is given as follows
\begin{equation}\label{eq:initial_re}
  f_0(x,k) = \frac{1}{\pi}\exp\left({-\frac{x^2}{2}}\right)\left[2\exp(-2(k-\pi)^2)-\exp({-2(k+\pi)^2})\right].
\end{equation}
And other parameters are set to be $-x_L=x_R=30$, $-k_{\min}=k_{\max}=5\pi$, $\Delta t = 10^{-4}$, $Q=20$, and the numerical solutions on the mesh $(M,N) = (41,512)$ provides the reference. For convenient for comparing, the non-splitting scheme here refers to our above-mentioned explicit conservative spectral method. Fig.~\ref{fig:converg_re} plots the convergence curves for both splitting and non-splitting schemes, which demonstrate clearly the spectral convergence against both $M$ and $N$ as we excepted, and also shows that they are almost identical. Consequently, we may not expect obvious difference in keeping the energy due to such high accuracy.
Actually, the variation of energy $\varepsilon_{\text{energy}}(t)$ on the finest mesh is no more than $1.8918\times 10^{-8}$ and $2.1852\times 10^{-8}$ until $t=10$ for the non-splitting and splitting schemes, respectively.
However, the difference on a coarse mesh may be evident as we already pointed out in Section~\ref{sec:splitting}, namely, the splitting scheme only keeps the mass while the non-splitting one keeps both the mass and the energy.  We choose a coarse mesh: $Q = 1$, $M=21$, $N=20$ and the splitting scheme \eqref{eq:split_num1} is used with a time step $\Delta t = 0.01$. Only after one time step, the variation of energy $\varepsilon_{\text{energy}}(\Delta t)$ is $4.6246\times 10^{-3}$ for the splitting scheme, but barely $7.1054\times 10^{-13}$ for the non-splitting one, 
and the variations of mass $\varepsilon_{\text{mass}}(\Delta t)$ for both are $4.4409\times 10^{-16}$.

\begin{figure}
   \includegraphics[width=0.5\textwidth,height=0.35\textwidth]{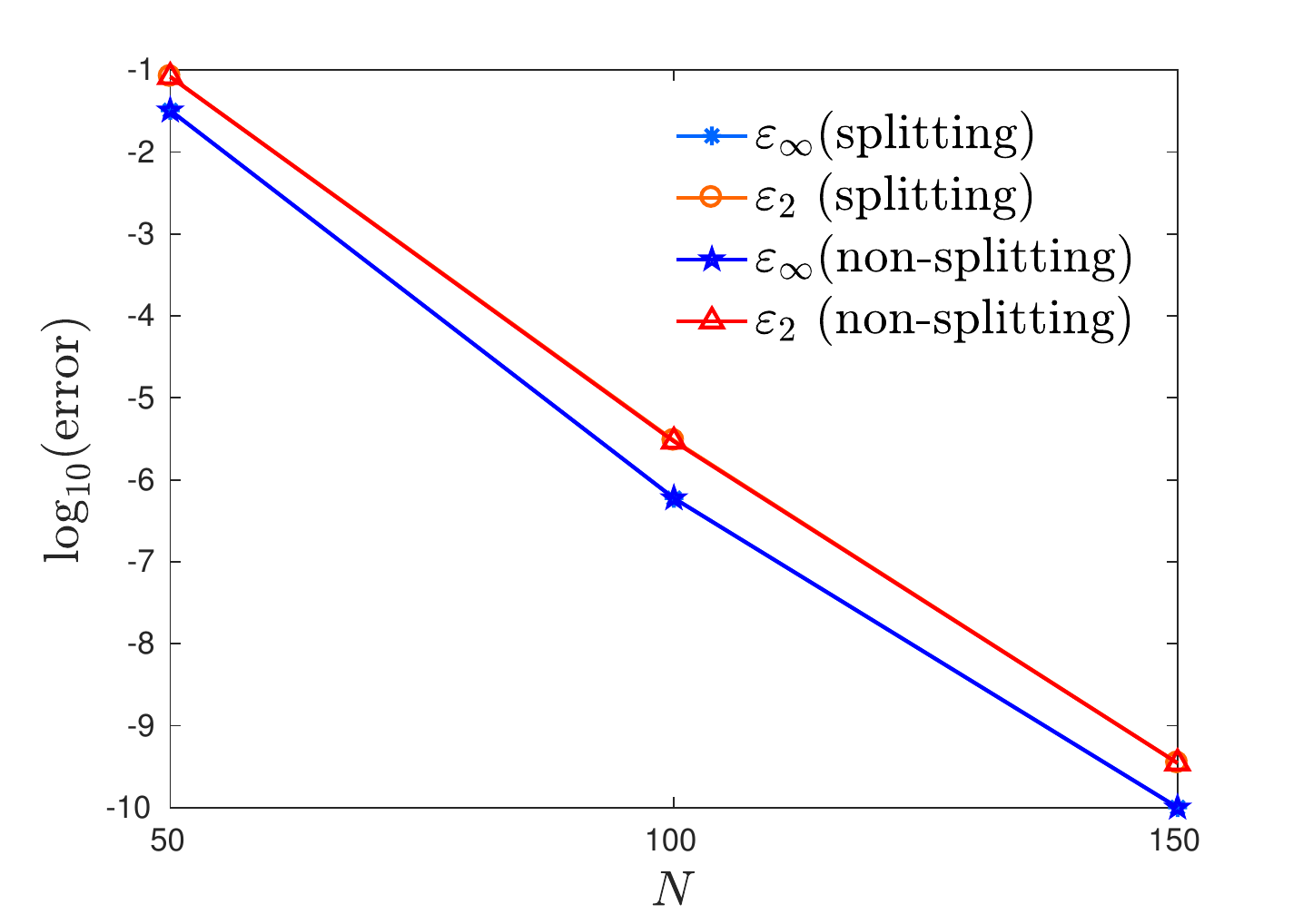}
   \includegraphics[width=0.5\textwidth,height=0.35\textwidth]{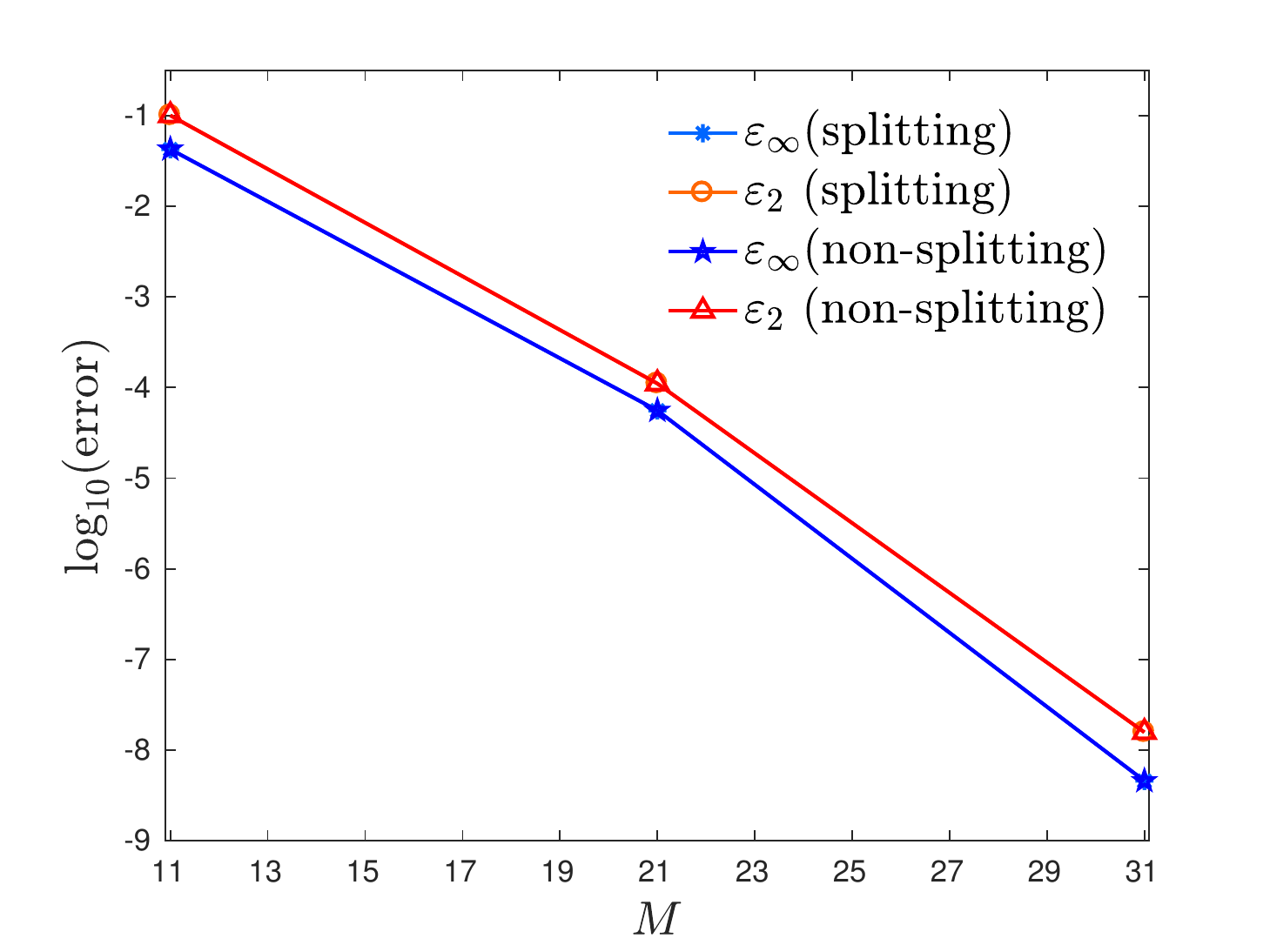}
   \caption{\small A rational fraction potential: The convergence rate with respect to $N$ (left) and $M$ (right). The spectral convergence in both $x$-space and $k$-space is obviously observed for both non-splitting and splitting schemes. Actually, their convergence curves are almost coincident. All the errors are measured at the instant $t=10$.}
   \label{fig:converg_re}
\end{figure}

\section{Conclusions and discussions}
\label{sec:5}

Using two equivalent forms of the pseudo-differential operator: the integral form and the Moyal expansion, 
we developed an explicit mass-and-energy-conserving spectral solver for the transient Wigner equation in the presence of a general class of unbounded potentials. Numerical experiments on several typical double-well systems demonstrate the spectral accuracy as well as the reliability of long time simulations. A direct spectral analysis of the resulting data demonstrates that the proposed solver accurately captures the energy level transitions. The uncertainty principle and the autocorrelation function are both investigated in the Wigner simulations of the quantum tunneling phenomena. We also showed that a simple operator scheme may keep the mass, but fails to conserve the energy. Now a project toward a mass-and-energy-conserving operator splitting method is still ongoing to fully explore its ability in handling different subproblems using different techniques.

\section*{Acknowledgments}
This research is supported by grants from the National Natural Science Foundation of China (Nos.~11471025, 11421101). 
Z. C. is also partially supported by Peking University Weng Hongwu original research fund (No. WHW201501). 
The authors are grateful to the useful discussions with Wei Cai, Jian Liu and Jing Shi.


\end{document}